\documentclass[preprint]{elsarticle}
\usepackage[english]{babel}
\usepackage[letterpaper,top=2cm,bottom=2cm,left=3cm,right=3cm,marginparwidth=1.75cm]{geometry}
\usepackage{amsmath, amsfonts,amssymb,amsthm}
\usepackage{graphicx}
\newtheorem{proposition}{Proposition}
\usepackage{indentfirst}
\usepackage{subcaption}
\usepackage{makecell}
\usepackage{xcolor}
\usepackage{tikz}
\usetikzlibrary{arrows.meta,calc}
\usepackage{pgfplots}
\usepackage{multirow}
\pgfplotsset{compat=1.18}
\usepackage[
colorlinks=true,%
breaklinks=true,%
linkcolor=blue,
urlcolor=blue,%
citecolor=blue,%
pdftitle={Uncertainty quantification of fatigue initiation life for powder bed fusion metal additive manufacturing}, 
pdfkeywords={},	
pdfauthor={Yulin Guo},		
bookmarksopen=false,
pdfpagemode=UseNone]{hyperref}

\newcommand{\bit}{\begin{itemize}}
\newcommand{\eit}{\end{itemize}}
\newcommand{\ben}{\begin{enumerate}}
\newcommand{\een}{\end{enumerate}}

\newcommand {\real} {\mathbb{R}}

\newcommand{\tl}{\ensuremath{\tilde{l}}}

\newcommand{\bC}{\ensuremath{\mathbf{C}}}
\newcommand{\bD}{\ensuremath{\mathbf{D}}}
\newcommand{\bE}{\ensuremath{\mathbf{E}}}
\newcommand{\bF}{\ensuremath{\mathbf{F}}}

\newcommand{\bI}{\ensuremath{\mathbf{I}}}

\newcommand{\bK}{\ensuremath{\mathbf{K}}}
\newcommand{\bL}{\ensuremath{\mathbf{L}}}

\newcommand{\bP}{\ensuremath{\mathbf{P}}}

\newcommand{\bR}{\ensuremath{\mathbf{R}}}
\newcommand{\bS}{\ensuremath{\mathbf{S}}}
\newcommand{\bT}{\ensuremath{\mathbf{T}}}

\newcommand{\bb}{\ensuremath{\mathbf{b}}}

\newcommand{\bm}{\ensuremath{\mathbf{m}}}
\newcommand{\bn}{\ensuremath{\mathbf{n}}}

\newcommand{\bp}{\ensuremath{\mathbf{p}}}

\newcommand{\br}{\ensuremath{\mathbf{r}}}

\newcommand{\bu}{\ensuremath{\mathbf{u}}}

\newcommand{\bz}{\ensuremath{\mathbf{z}}}

\newcommand {\blambda} {\mbox{\boldmath $\lambda$}}

\newcommand{\bsigma}{\ensuremath{\boldsymbol{\sigma}}}

\newcommand{\cC}{\ensuremath{\mathcal{C}}}

\newcommand{\ca}{\ensuremath{\mathcal{a}}}

\newcommand{\cs}{\ensuremath{\mathcal{s}}}

\begin{document}

\begin{frontmatter}
		
\title{Uncertainty quantification of fatigue initiation life for powder bed fusion metal additive manufacturing}

\author[ucsd]{{Yulin Guo}\corref{cor1}}
\ead{yug054@ucsd.edu}

\author[umich]{Veera Sundararaghavan}
\ead{veeras@umich.edu}
		
\author[ucsd]{Boris Kramer}
\ead{bmkramer@ucsd.edu}

\cortext[cor1]{Corresponding author}
                
\address[ucsd]{Department of Mechanical and Aerospace Engineering, University of California San Diego, CA, United States}

\address[umich]{Department of Aerospace Engineering, University of Michigan Ann Arbor, MI, United States}

\begin{abstract}
    Predicting fatigue life with quantified uncertainties is essential for the qualification of critical components produced by laser-based powder bed fusion additive manufacturing. We present a framework that propagates microstructure and defect uncertainties directly to a fatigue initiation life distribution for a specific part. In particular, microstructure and defect characterizations are obtained from electron backscatter diffraction and micro-computed tomography scan data, which in turn inform three physics-based simulations yielding the fatigue-affecting quantities: the elastic energy release rate, the surface energy along the crack path, and the fatigue indicator parameter. Accounting for the uncertainties in these quantities and the high correlations among them due to the shared underlying microstructure, we derive a closed-form probability density function for the fatigue initiation life. This provides an analytical distribution instead of conservative deterministic predictions and enables more informed decision making for the qualification and deployment of additively manufactured components. Applying the framework to 316L stainless steel parts produced by an EOS M290 laser powder bed fusion machine, we find that both grain sizes and void distributions influence the fatigue initiation life distribution. Specifically, for a fixed total void volume fraction, larger grain sizes cause a marginal reduction in fatigue initiation life, and a population of many small voids is more favorable for fatigue life than fewer, larger voids of equivalent total volume.
\end{abstract}

\end{frontmatter}

\section{Introduction}\label{section:intro}


Additive manufacturing (AM), the process of joining materials layer by layer to make parts from 3D model data, has been a popular prototyping tool~\cite{gardan2017additive}. The American Society for Testing and Materials classifies AM processes into seven categories: binder jetting, directed energy deposition, material extrusion, material jetting, powder bed fusion, sheet lamination, and vat photopolymerization~\cite{cdi_astm_standards_111828}. A variety of materials, including metals, are available to provide the unique attributes needed to satisfy the engineering requirements of AM parts in aerospace, automotive, medical devices, and other industries~\cite{shapiro2016additive, vasco2021additive, seoane2021semi, murr2018additive, wang2024uncertainty}. Several additive manufacturing processes are available for metallic materials, with each process utilizing specific material forms: metal sheets for sheet lamination, metal rods or wires for material extrusion, metal wires for directed energy deposition, and metal powders for binder jetting and powder bed fusion~\cite{konda2017additive}. Selective laser melting is a powder bed fusion technique and is typically used for metallic powders. It operates with a cold powder bed inside a closed chamber filled with inert gases such as N$_2$ or Ar, which prevent oxidation of the metal powder during melting~\cite{konda2017additive}. 

Compared to traditional subtractive and formative manufacturing methodologies, such as milling and injection molding, AM has the ability to make highly complex geometric designs. Combined with its high speed and low labor cost, AM has gained popularity beyond modeling, prototyping, and tooling and has been used to produce complex, small-batch parts that are impractical to manufacture using conventional methods~\cite{guo2025risk}. The challenge in additive manufacturing stems from inherent uncertainties in both material properties and process conditions. Compared to traditional manufacturing, these variations are harder to control and can lead to inconsistent part quality. As a result, the final component may contain defects introduced during printing, which increase variability in performance and reliability.


The selective laser melting process for metal involves over 50 parameters, including laser power, scanning speed, hatch distance, and overlaps~\cite{spears2016process, konda2017additive}. Additively manufactured parts can exhibit unique microstructures due to thermal gradients, including porosity, grain features, and melt pool boundaries between deposit layers and scanning tracks~\cite{ronneberg2020revealing}. Different machines and process parameters can produce parts with grain lengths ranging from 10 to 500 $\mu$m. In thicker parts, microstructures with columnar orientations along the build direction can be observed, which promote twinning-induced plasticity~\cite{wang2018microstructure}. At the grain level, cellular subgrain structures are oriented along the grain solidification direction with widths ranging from 0.3 to 1.0 $\mu$m due to variations in primary dendrite arm spacing between adjacent grains~\cite{afkhami2021effects}. Subgrain sizes are also affected by process parameters. For example, a reduced hatch distance can lead to thicker cellular structures, and a faster cooling rate can lead to smaller primary dendrite arm spacing and higher strength \cite{eliasu2021effect} based on the Hall--Petch relationship \cite{hall1951deformation, petch1958ductile, hansen2004hall}. Grain features can be observed by optical microscopy and scanning electron microscopy~\cite{ronneberg2020revealing}. Grain size measurements depend on both the measurement plane and the measurement method~\cite{riabov2021investigation}. Electron backscatter diffraction (EBSD) inverse pole figure maps can be obtained for different process parameter sets using various machines, such as 3D Systems DMP 350, EOS M290, AddUp FormUp 350, and Renishaw AM 400~\cite{schreiber2025rationalizing}.


Under cyclic loading conditions, cracks in additively manufactured metallic parts may initiate from preexisting voids and propagate, leading to fatigue failures. The ratio of fatigue limit to tensile strength has been reported as 0.12--0.24 for 316L stainless steel manufactured by the laser-based powder bed fusion technique, compared to 0.40--0.60 for wrought counterparts~\cite{hamada2023enhancement}. This indicates reduced fatigue performance in 316L stainless steel parts produced by the powder bed fusion technique~\cite{afkhami2021effects}. The loading direction relative to a defect's major axis plays a key role in fatigue life \cite{wang2018microstructure}. In the presence of lack-of-fusion defects, horizontal samples tend to last longer under fatigue loading conditions. The higher ductility of columnar texture may also contribute to longer fatigue life due to improved energy absorption, and the effects of loading direction with respect to lack-of-fusion defects are discussed in~\cite{richardsen2023effect}. In principle, while the fatigue performance of metal AM parts may be improved by minimizing defect formation during~printing, the as-printed microstructure, including metallic grains and nonmetallic inclusions, is the primary driver of fatigue initiation life. Therefore, a quantitative understanding of its influence is essential for reliable fatigue life prediction.
    

Recent efforts to predict the fatigue life of printed metal parts often combine simulations and experimental testing. When considering uncertainties in quantities that affect fatigue life, such as the size of nonmetallic inclusions in steel parts, the common practice is to fit extreme value distributions to limited measurements~\cite{astm_standards_2283}. The authors in~\cite{ponticelli2022experimental} consider different build orientations in 316L stainless steel and conduct reverse bending fatigue tests. They use scanning electron micrographs to analyze fracture surfaces and relate the results to volumetric energy density. The work~\cite{aiza2025effects} considers different build orientations (11 setups between 0$^{\circ}$ and 90$^{\circ}$) of Ti6Al4V, measures mechanical properties (hardness, tensile strength, fatigue strength, and fracture toughness) as well as wear and corrosion properties, and uses micro-computed tomography to characterize pores and EBSD to examine post-heat-treatment microstructures. The authors in~\cite{douglas2024influence} consider different sample shapes, surface finishes, and build directions for 316L stainless steel AM parts to investigate low-cycle fatigue behavior. They use optical microscopy to identify defects and EBSD to characterize the microstructure. The above works focus on a fixed set of build conditions and rely on costly experiments to qualitatively relate fatigue life to process parameters. There are also several works on neural networks to predict fatigue behaviors. The authors in~\cite{tognan2024bayesian} use a defect-based Bayesian neural network to evaluate the fatigue endurance of selective laser melted AlSi10Mg parts. However, the network only predicts whether a specimen attains a given fatigue life subjected to a prescribed tensile fatigue load and only a limited number of samples with certain defect traits are used for training and testing. The authors in~\cite{wang2025multi} use a multi-fidelity neural network with defects as constraints to predict the fatigue performance of Ti6Al4V parts, which are produced with a few sets of process parameters and tested under a few conditions. These approaches lack the ability to reliably predict the fatigue life of parts built with untested printing parameter sets. Another approach for fatigue life prediction of additively manufactured parts is multi-physics modeling~\cite{yan2018integrated}. In this framework, the authors combines models for the manufacturing process, material structure formation, and mechanical response. However, the predicted fatigue life is calculated empirically with a nominally defect-free baseline and the uncertainty associated with the prediction is not explicitly quantified.

 
We develop methods to quantify the uncertainty in fatigue initiation life for as-built parts by incorporating three fatigue-affecting quantities, namely the elastic energy release rate, the surface energy along the crack path, and the fatigue indicator parameter, along with their associated uncertainties. These quantities are based on the microstructure, including both the metallic grains and the defects, of 316L stainless steel parts produced by the laser-based powder bed fusion technique. We account for correlations among these three quantities, as they are derived from the same microstructure data and finite element mesh. The resulting fatigue initiation life is not merely a point estimate but a full probability distribution, which provides actionable guidance on how the part can be safely utilized.
    
The novel contributions of this work are: (1) We establish a {very fast} forward uncertainty propagation framework to quantify the uncertainty in the fatigue initiation life of additively manufactured parts directly from as-printed microstructure features, including both the metallic grains and defects; (2) we combine the microstructure data of the printed parts, first-principles-based physical model simulations, and statistical model selection to build the predictive framework; (3) to reduce the time to prediction, we accelerate the linear elastic finite element simulation {by developing a} two-fidelity scheme with a lower-fidelity preconditioner and warm-start initial guess. The simulation results are used to calculate the elastic energy release rate, one of the three fatigue-affecting quantities; (4) by incorporating spherical voids in metallic polycrystalline grains and considering proper material boundary interactions, we extend the 3D graph-theoretical energy-minimization model, MicroFract3D, to predict the crack path and obtain the crack-path energy based on local resistance analysis; (5) with a block maxima approach and physical models under a shared microstructure, we derive a closed-form analytical probability density function for fatigue initiation life under correlated variables. By contrast, the traditional approach assumes a generic distribution such as a Gumbel or Weibull distribution for the governing~variables; {(6) we validated our closed-form PDF for the predicted fatigue life against fatigue life experimental test, and find that the number of cycles to failure is within our proposed closed-form PDF corresponding to the printing parameters.}

The rest of this paper is organized as follows. We first introduce the physics that governs fatigue behavior, the fatigue-affecting quantities, and the calculation of fatigue initiation life in Section~\ref{sec: fatigue-life}. We then discuss the simulations for obtaining these fatigue-affecting quantities in Section~\ref{sec: simulation}. In Section~\ref{section:uncertainty-initiation-life}, we detail the uncertainty in fatigue initiation life based on its functional form. Specifically, we derive a closed-form probability density function for the fatigue initiation life, considering the correlations among and uncertainties in the fatigue-affecting quantities. In Section~\ref{sec: num}, we provide numerical results. Based on the microstructure, including metallic grains and voids, of 316L stainless steel AM parts, we show the propagation of uncertainties from fatigue-affecting quantities to the fatigue initiation life. Finally, we conclude our work and identify future research directions in Section~\ref{sec: con}.

\section{Fatigue life modeling}\label{sec: fatigue-life}

The fatigue life of an additively manufactured metallic part under cyclic loading is governed by the initiation and early growth of cracks, which are influenced by complex microstructure-driven plasticity. To predict a specific part's fatigue life, we use three variables: the surface energy $\gamma$, the elastic energy release rate $G$, and the plasticity related fatigue indicator parameter (FIP) - the stored plastic work $\omega^p$ \cite{fine2007model}. We first introduce the physical background and models associated with initiation life in Sections~\ref{section:fatigue-life-cyclic-loading} and \ref{section:three-models}, followed by the formulation of the physical principles connecting these variables to fatigue crack initiation. As detailed below, we use J/m$^2$ for all three quantities.

\subsection{Fatigue life under cyclic loading}\label{section:fatigue-life-cyclic-loading}

Under cyclic loading, metal components experience progressive damage that can lead to fatigue failure. This process can be divided into two stages: crack initiation and crack propagation. In the first stage, fatigue damage accumulates at the microstructural level. This leads to the nucleation of micro cracks from small voids, which subsequently coalesce into macro cracks~\cite{griffith1921vi}. In the second stage, macro cracks propagate through the bulk material under cyclic loading. Crack propagation is governed by Paris-type laws~\cite{paris1963critical}. When the component can no longer sustain the load, sudden failure occurs. In high cycle fatigue of additively manufactured 316L stainless steel with reduced surface roughness, the crack initiation phase dominates the total fatigue life, accounting for 70\%-90\% of the total life \cite{avanzini2022fatigue}. These fatigue cracks nucleate from pre-existing printing defects, such as trapped pores, lack-of-fusion pores, and keyhole defects due to local overheating. 

Crack initiation corresponds to the separation of neighboring atoms to form two new free surfaces, which requires work equal to the atomic bond energy. For a crack of length $a$, we denote the thickness of the part at the crack site as $B$. Microscopically, the crack initiation creates two new rectangular surfaces with length $a$ and width $B$. Thus the total atomic bond-breaking energy is $E_{\text{bond}} = 2\gamma_s a B$, where $\gamma_s$ is the bond-breaking energy per unit area. For brittle materials, it was shown in~\cite{griffith1921vi} that a crack will advance when the released elastic strain energy equals the surface energy of the newly created crack surfaces. In metallic materials, however, plastic deformation occurs near the crack tip. This dissipates more energy than the cleavage of atomic bonds alone. This extension~\cite{orowan1949fracture,irwin1957analysis} defines an effective surface energy $\gamma$ that incorporates both the cleavage energy and the plastic work of deformation. Under cyclic loading, this local plastic work accumulates in the form of dislocation structures within the grain microstructure. A fatigue crack initiates when the sum of the released elastic strain energy and the accumulated stored plastic strain energy of dislocations overcomes the local effective surface energy barrier~\cite{tanaka1981dislocation,zhao2014energy}.

\subsection{The three models associated with fatigue initiation life estimation}\label{section:three-models}

We are interested in the uncertainty in the fatigue initiation life of an additively manufactured 316L stainless steel part produced using the selective laser melting technique. As described in Section~\ref{section:intro}, the metallic powder melts and solidifies to form a part. The thermal-mechanical behavior during the printing process results in a complex microstructure in the part. To quantify the uncertainty in its fatigue initiation life, {we use three models to obtain fatigue-affecting quantities, see Figure~\ref{fig:three-ingridients}. In Sections~\ref{section:uncertainty-initiation-life}~and~\ref{sec: num}, with simulation results using printed specimen-informed microstructures, we demonstrate the uncertainty quantification in these quantities and the fatigue initiation life.}

\begin{figure}[h]
    \centering
    \input{fig-three-ingredients-R1}
    \caption{We use three models to find the fatigue-affecting quantities $G$, $\gamma$, and $\omega^p$. {The metallic grains and voids shown here are generated based on geometry properties measured by electron backscatter diffraction and computed tomography scans of printed specimens, see Section~\ref{sec: num}.}}
    \label{fig:three-ingridients}
\end{figure}

The three models and their corresponding outputs are as follows: (1) We compute the elastic energy release rate $G$ for spherical voids using a finite element model combined with an analytical solution from linear elastic fracture mechanics. (2) We find the minimum-energy crack path to obtain the surface energy $\gamma$ along the crack surfaces using a graph-based methodology. Fatigue crack paths are significantly influenced by the microstructure. For instance, microstructures with columnar grains aligned parallel to the loading direction can exhibit enhanced fatigue strength compared to those with equiaxed grains, due to the directional resistance to crack propagation. Therefore, we use microstructure models of metallic grains and voids in the simulation. (3) At the microstructure level, we use crystal plasticity finite element simulations to evaluate the accumulation of plastic work near the flaws. Based on the simulation results, we calculate the stored plastic work $\omega^p$ to quantify the cyclic stored plastic work, which correlates with the material's fatigue resistance.

Next, we derive the fatigue initiation life as a function of these three quantities. In the initiation event, we assume that a pre-existing crack begins to extend, which results in an incremental increase in the crack area, denoted as $A$. Crack formation leads to the release of pent-up elastic energy from the applied load. We define the elastic energy release rate $G$ in terms of the stored elastic energy $U^e$ 
    \begin{equation}
        G := -\frac{\mathrm{d} U^e}{\mathrm{d} A}.
        \label{eq:G_energy_release_rate_U}
    \end{equation}
In Mode-I loading, $G = K_{\text{max}}^2/E'$, where $K_{\text{max}}$ is the maximum stress intensity factor and $E'$ is the effective modulus. For plane strain, $E' = E/(1-\nu^2)$ where $\nu$ is Poisson's ratio. We compute $K_{\text{max}}$ based on the spherical defect diameter and the elastic stress field from the linear elastic finite element~simulations.

Crack formation releases the stored plastic work in the form of dislocations near the crack nuclei. The energy from loading is stored continuously by dislocations over the number of fatigue loading cycles. The decrease in the stored energy due to the increase in crack area is given by
    \begin{equation}
        \frac{\mathrm{d} \omega^p}{\mathrm{d} A} =: - N\omega^p.
        \label{eq:omega_p_N_W}
    \end{equation}
The negative sign denotes that this energy is released when the crack nucleates. The FIP $\omega^p$ can be viewed as the energy stored per cycle per unit area, which is obtained by applying a full loading cycle to the microstructure. In the crystal plasticity finite element simulation, we calculate the stored plastic work by integrating the plastic work rate over the entire loading cycle on a representative~microstructure.

The energy barrier to cracking is related to the increase in energy due to the formation of new surfaces, quantified per unit area as $\gamma$. We can write the change in surface energy with respect to crack area as:
    \begin{equation}
        \frac{\mathrm{d}W^s}{\mathrm{d}A} =: 2\gamma,
        \label{eq:gamma_surface_energy}
    \end{equation}
where the factor 2 accounts for the two surfaces created by the crack. Surface energy is material- and microstructure-dependent because it is related to the impurities and the strength of internal surfaces, such as slip planes and grain boundaries. We model these quantities in the graph-based crack path simulation to calculate the surface energy.

A crack initiates when the following stationarity~\cite{fine2007model} is achieved, as modeled by
    \begin{equation}
        \frac{\mathrm{d}(U^e+W^p+W^s)}{\mathrm{d}A}=0.
        \label{eq:crack_initiation_stationarity}
    \end{equation}
We obtain the number of cycles to crack initiation, $N$, by substituting equations~\eqref{eq:G_energy_release_rate_U}~to~\eqref{eq:gamma_surface_energy} into eq.~\eqref{eq:crack_initiation_stationarity}:
    \begin{equation}
        N = \frac{2\gamma-G}{\omega^p}.
        \label{eq:fatigue_life}
    \end{equation}

Next, we discuss the three models in detail in Section~\ref{sec: simulation}. In most cases, the parameters $\gamma$, $G$, and $\omega^p$ are uncertain quantities such that $N$ becomes a random variable. In this case, we derive the closed-form solution for the distribution of $N$ in Section~\ref{section:uncertainty-initiation-life}, where we consider the uncertainties in the three quantities and correlations among them.

\section{Simulation details}\label{sec: simulation}

We use eq.~\eqref{eq:fatigue_life} to determine the fatigue initiation life of a 3D-printed 316L stainless steel part using quantities obtained from three simulations. To estimate the initiation life $N$, we require three physical quantities: the elastic energy release rate $G$, the surface energy $\gamma$, and the peak plastic work $\omega^p$. We discuss the computational methodologies used to extract these quantities. Specifically, we present the details of the linear elastic finite element analysis, the graph-theoretic crack path simulation, and the crystal plasticity finite element analysis in Sections~\ref{section:simulation-LEFM-G}, \ref{section:simulation-MicroFract3D}, and \ref{section:simulation-CPFE},~respectively.

\subsection{Elastic energy release rate $G$ by linear elastic finite element analysis}\label{section:simulation-LEFM-G}

We detail the computation of the elastic energy release rate $G$, beginning by outlining the linear elastic voxel-based solver ELAS3D {originally described in NISTIR-6269~\cite{garboczi1998finite}} in Section~\ref{section:ELAS3D}. We then describe how we determine the maximum stress intensity factor $K$ and calculate the resulting energy release rate $G$ in Section~\ref{section:SIF-G} {based on the simulation results. To accelerate the full elastic stress field computation, we introduce a two-resolution scheme and discuss it in Section~\ref{section:two-resolution-scheme}}.


\subsubsection{Linear elastic simulation using ELAS3D}\label{section:ELAS3D}

To obtain the linear elastic energy release rate $G$, we first compute the full elastic stress field of the microstructure containing metallic grains and spherical voids. For this purpose, we run linear elastic simulations {using a modified version of the ELAS3D code}. 

Each simulation domain is represented by a 3D digital image of cubic voxels where the voxelized microstructure is generated using Voronoi tessellation. For metallic voxels, we assign a single-crystal stiffness tensor $\mathbf{C} \in \mathbb{R}^{6\times6}$ rotated to the local grain orientation. For void voxels, we assign a near-zero stiffness so that they carry negligible load. To discretize the domain, each voxel is treated as a trilinear finite element with 8 corner nodes. A standard node labeling scheme is implemented to reference the element degrees of freedom consistently across the grid. We select the void stiffness value to ensure numerical stability in simulations, as detailed in Section~\ref{sec: num}. 
    
Once the voxelized microstructure is defined, the corresponding elasticity problem is set up using a variational principle. The solver seeks the nodal displacements $w_{\text{rp}}$ that minimize the total elastic energy $U$ in the entire simulation domain, satisfying $\partial U/\partial w_{\text{rp}} = 0$ for all degrees of freedom, where the subscript $\text{r} \in \{1,\dots,8\}$ labels the node and $\text{p} \in \{x,y,z\}$ labels the translational displacement direction. To express this energy, we first use Voigt notation to store the six independent strain components as
    \begin{equation}
        \boldsymbol{\epsilon}=(\epsilon_{xx},\epsilon_{yy},\epsilon_{zz},\epsilon_{yz},\epsilon_{xz},\epsilon_{xy}).
        \label{eq:Voigt-strain}
    \end{equation}
The voxel small-strain linear elastic energy is thus written as
    \begin{equation}
        U^e_{\text{voxel}} = \frac{1}{2}\int_{\text{voxel}} \epsilon_{\alpha} C_{\alpha\beta} \epsilon_{\beta}\, \mathrm{d}V,
        \label{eq:voxel-elastic-energy}
    \end{equation}
where $\alpha,\beta \in \{1,2,\ldots,6\}$ are Voigt indices ($\epsilon_1 \Leftrightarrow \epsilon_{xx}$, $\epsilon_2 \Leftrightarrow \epsilon_{yy}$, $\ldots$ according to eq.~\eqref{eq:Voigt-strain}), and $C_{\alpha\beta}$ are the components of the single-crystal stiffness tensor $\bC$. We first sum the integrand over the Voigt indices locally, and then sum over all voxels to obtain the total energy $U$ of the simulation domain. 

To discretize the variational formulation, we interpolate the displacement inside each voxel using the voxel shape functions and write each component of $\bu(x,y,z)$ as $u_{\text{p}}(x,y,z)=N_{\text{r}}(x,y,z)\, w_{\text{rp}}$, where $N_{\text{r}}(x,y,z)$ is the trilinear shape function that maps the 24 displacement degrees of freedom to the displacement component $u_{\text{p}}$ at an interior point $(x,y,z)$ within the voxel. The local strain field is then computed from displacement gradients by applying a Voigt-form derivative operator to the interpolated displacement field. This allows us to express the local strain as a linear map of nodal displacements, $\epsilon_{\alpha}(x,y,z)=S_{\alpha \text{rp}}(x,y,z)\, w_{\text{rp}}$, where $S_{\alpha \text{rp}}$ is the strain-displacement operator that takes the spatial derivatives of the shape functions and maps the nodal displacement component $w_{\text{rp}}$ to the local strain component $\epsilon_{\alpha}$ at a point inside the voxel. Using this representation, each voxel stiffness matrix $\bD \in \real^{24\times 24}$ is assembled by integrating the strain-displacement operator against the local moduli, yielding components $D_{\text{rp},\text{sq}} = \int_{\text{voxel}} S_{\alpha \text{rp}}^{\top} C_{\alpha\beta} S_{\beta \text{sq}}\, \mathrm{d}V$. The voxel energy is then written in quadratic form using nodal unknowns, where the first pair of subscripts, $\text{r}$ and $\text{p}$, indexes the component's row and the second pair, $\text{s}$ and $\text{q}$, indexes the column.
    
We evaluate the voxel integrals using Simpson's rule in the implementation. Note that the integrands are at most quadratic. To efficiently manage the degrees of freedom and speed up computation, we implement three strategies in the simulation: first, the 3D grid indices $(i,j,k)$ are mapped to a single global node label $m=(k-1)n_x n_y + (j-1)n_x + i$, where $n_x, n_y, n_z$ are image dimensions. We can then store the microstructure as a 1D array and assemble global contributions efficiently; second, a fixed neighbor labeling is used by storing the 27-node neighborhood in a 1D table to gather the nodes that couple through the surrounding voxels. We also apply periodic wrap-around for neighbors that fall outside the domain to satisfy the periodic boundary conditions; and third, we compute the global energy gradient without storing the full global stiffness matrix explicitly; instead, we rebuild the needed matrix–vector products from local voxel stiffness matrices and neighbor connectivity to reduce memory use.
    
To model the bulk material response, we use periodic boundary conditions for each simulation domain so that each cube behaves as a repeating unit cell under a prescribed macroscopic loading. We apply a uniaxial tensile strain $\epsilon_{zz}$ with all other macroscopic strain components set to zero, which corresponds nominally to Mode-I loading. The macroscopic strain ${\bE}$ is applied through periodic displacement jumps at the unit-cell boundaries, and we incorporate these jumps through an additive correction in the boundary voxel energy contributions. 

With these boundary conditions implemented, the global energy functional remains quadratic in the unknown nodal displacements, and we interpret its stationary point as the discrete solution for the displacement field under the prescribed macroscopic strain. The stationary condition $\partial U / \partial w_{\text{rp}} = 0$ assembles into a global linear~system
    \begin{equation}
        \bK \bu = -\bb,
        \label{eq:elas3D-linear-system}
    \end{equation}    
where $\bK \in \real^{n \times n}$ is the global stiffness matrix, $\bu \in \real^{n}$ is the vector of unknown nodal displacements, $-\bb \in \real^{n}$ is the external forcing vector arising from the applied macroscopic strain, and $n = 3\,(n_x+1)(n_y+1)(n_z+1)$ is the total number of degrees of freedom. We solve this system by using a preconditioned conjugate-gradient solver and stopping the iterations once the relative residual norm falls below a specific tolerance. We discuss the preconditioner and the warm-start initial guess in detail in Section~\ref{section:two-resolution-scheme}.
    
Finally, we post-process the converged displacements to compute voxel-averaged strains and stresses. The {ELAS3D} solver outputs the full six-component stress tensor $(\sigma_{xx}, \sigma_{yy}, \sigma_{zz}, \sigma_{xz}, \sigma_{yz}, \sigma_{xy})$ at each voxel, as well as the volume-averaged macroscopic stress $\bsigma^{\mathrm{avg}}$. 


\subsubsection{Stress intensity factor $K$ and Elastic energy release rate $G$}\label{section:SIF-G}

Using the volume-averaged stress $\bsigma^{\mathrm{avg}}$ computed in the previous step, we can determine the driving force for fatigue crack nucleation at the printing defects. {Specifically, we model each spherical void of radius $a$ as an embedded penny-shaped crack of the same radius, loaded by the far-field stress.} This is an analytical approach that avoids extracting stress intensity factors (SIFs) from near-tip stress fields, which are mesh-dependent and contaminated by stress concentrations near the void boundary~\cite{sneddon1946distribution}.

For the far-field stress, we use the volume-averaged stress $\bsigma^{\mathrm{avg}}$ over the simulation domain. This macroscopic stress serves as a far-field proxy because it is insensitive to mesh refinement near the void, and the simulation domain is much larger than the void so that the volume average is not significantly perturbed by the near-void stress concentration. In contrast, extracting SIFs by fitting the near-tip $\sigma_{ij}(r,\theta) \approx K/\sqrt{2\pi r}\, f_{ij}(\theta)$ asymptotic form would require careful mesh-convergence studies near the void boundary, which the voxelized grid does not naturally support.

Based on this macroscopic stress state, we apply the Sneddon solution for an embedded penny-shaped crack of radius $a$ in an infinite elastic body~\cite{sneddon1946distribution}, and compute the SIFs under the Mode-I~loading:
    \begin{equation}
        K_I   = \frac{2}{\pi}\, \sigma_{zz}^{\mathrm{remote}}\, \sqrt{\pi a},
        \label{eq:K_I_Sneddon}
    \end{equation}
where $\sigma_{zz}^{\mathrm{remote}}$ is the relevant component of the volume-averaged stress $\bsigma^{\mathrm{avg}}$ corresponding to the tensile loading direction. Here, Mode-I corresponds to the crack opening driven by the applied uniaxial tensile strain $\epsilon_{zz}$. We compute $K_I$ at every spherical void and then compute $G$ as described in Section~\ref{section:three-models}. The void with the largest $G$ is identified as the most critical flaw, representing the flaw with the highest crack-driving force under the applied loading.

Although the SIF formula in~\eqref{eq:K_I_Sneddon} is derived for an embedded penny-shaped crack, whereas the actual embedded defects are spherical voids rather than flat, circular cracks, this model provides a useful and conservative approximation. A spherical cavity is a 3D volumetric defect with a smooth surface and a finite stress concentration factor (about 2.045 at the equator for $\nu \approx 0.3$, from the Goodier solution~\cite{goodier1933concentration, pilkey2020peterson}), in contrast to a crack which has a singular $1/\sqrt{r}$ stress field at its tip. {By equating the void radius to the penny-shaped crack radius, we treat the void as if a crack of that size were already present. This is a conservative estimate of the stress intensity factor because a flat crack of radius $a$ causes more stress concentration than a spherical void of the same radius.} Stress intensity factors for spherical voids have been studied for ceramic materials~\cite{baratta1978stress} and under certain initial crack assumptions~\cite{zhang2024stress}. {Ultimately, the penny-shaped crack approximation avoids complex stress concentration factor calculations and initial crack length assumptions while providing a conservative bound on $K_I$.}

\subsubsection{Two-fidelity solver}\label{section:two-resolution-scheme}

As introduced in eq.~\eqref{eq:elas3D-linear-system} in Section~\ref{section:ELAS3D}, the linear elastic problem reduces to a large linear system $\bK \bu = -\bb$ with $n$ degrees of freedom. {To speed up the computation, we solve this system by developing a two-resolution scheme: we first run the simulation on a coarse mesh and obtain the converged coarse displacement field~$\bu^{\text{c}}$, which we then prolongate to~$\bu^{\text{f}}$ on a fine mesh to use as a warm-start initial guess. The length of each cubic voxel in the coarse mesh is twice the length of each cubic voxel in the fine mesh.}

{In the prolongation step, we first subtract the homogeneous macroscopic deformation due to the applied strain $\bE$ from the coarse displacement field~$\bu_*^{\text{c}}$ first.} We only prolongate the resulting local, microscopic displacement deviations~$\tilde{\bu}_*^{\text{c}}$. We denote each displacement component by $u_{\text{p}}$ and write the microscopic deviation component as
    \begin{equation}
        \tilde{u}^{\text{c}}_{\text{p}}(i^{\text{}}, j^{\text{}}, k^{\text{}}) = u_{\text{p}}^{\text{}}(i^{\text{}}, j^{\text{}}, k^{\text{}}) - \left( X_i^{\text{}} E_{\text{p}x} + Y_j^{\text{}} E_{\text{p}y} + Z_k^{\text{}} E_{\text{p}z} \right),
    \end{equation}
where $i,j,k$ are the 0-based grid indices of the node on the coarse mesh, $X_i^{\text{}}, Y_i^{\text{}}, Z_i^{\text{}}$ are the corresponding node coordinates, $\text{p} \in \{x,y,z\}$ represents the displacement component direction, and $E_{\text{pq}}$ represents the components of the applied macroscopic strain~$\bE$.

{Next, we calculate the corresponding nodal displacement field~$\tilde{\bu}^{\text{f}}$ on the fine mesh. For coincident nodes, where the fine mesh node coordinates match the coarse coordinates, we copy the~$\tilde{\bu}$ value directly. For other nodes, we apply trilinear interpolation using the 8 neighboring coarse nodes to obtain~$\tilde{u}_{\text{p}}^{\text{f}}$.} We express the trilinear interpolation as:
    \begin{equation}
        \begin{aligned}
        &\tilde{u}^{\text{f}}_{\text{p}}(i^{\text{f}}, j^{\text{f}}, k^{\text{f}}) = \sum_{a,b,c \in \{0, 1\}} W_{abc}(w_x, w_y, w_z) \tilde{u}^{\text{c}}_{\text{p}}(i+a, j+b, k+c), \\
        \end{aligned}
        \label{eq:triliner-fine}
    \end{equation}
where $W_{abc}(w_x, w_y, w_z) = w_x^a (1-w_x)^{1-a} w_y^b (1-w_y)^{1-b} w_z^c (1-w_z)^{1-c}$ is the trilinear interpolation weight, $i, j, k$ are 0-based grid indices of the node on the coarse mesh, $a, b, c \in \{0, 1\}$ are binary offsets indexing the eight corners of the cell, $i^{\text{f}}, j^{\text{f}}, k^{\text{f}}$ are the grid indices of the corresponding node on the fine mesh, and $w_x, w_y, w_z \in \{0, 0.5\}$ are coordinate weights. Since the fine grid spacing is exactly half of the coarse grid spacing, each coordinate weight is either $0$ when the fine node lies on a coarse grid plane in that direction, or $0.5$ when it lies midway between two coarse grid planes.
    
{To ensure consistency with the periodic boundary conditions, we apply periodic wrap-around for nodes on the boundaries of the simulation domain. Additionally, we apply a zero-mean correction to $\tilde{u}^{\text{f}}_{\text{p}}$ obtained in eq.~\eqref{eq:triliner-fine} by subtracting its spatial average to prevent numerical drift in the solver. After this interpolation and zero-mean correction, we add the homogeneous macroscopic affine deformation component back to $\tilde{\bu}^{f}$ to obtain~$\bu^{\text{f}}$, which serves as the warm-start fine-mesh displacement~field.}

For the simulation on the fine mesh with $n_f$ degrees of freedom, {we further expedite the solution of eq.~\eqref{eq:elas3D-linear-system} by preconditioning the stiffness matrix.} Specifically, we solve the system using a conjugate gradient solver where the residual at iteration $k$ is expressed as $\br_k=\bK \bu_k + \bb$. We stop the solver when the relative residual norm $\sqrt{\langle \br_k, \br_k \rangle/\langle \br_0, \br_0 \rangle} \le 10^{-2}$, where $\langle \cdot, \cdot \rangle$ denotes the Euclidean inner product. We use a preconditioner that approximates the inverse of the stiffness matrix block-diagonally. {To minimize memory consumption, we never assemble the global preconditioning matrix; instead, we compute the correction locally for each node.} For node $m$, we denote the $3\times 3$ diagonal block of the global stiffness matrix $\bK$ as $\bK_m$. At any iteration, for the residual $\br_m$ at node $m$, we calculate a preconditioned residual correction as $\bz_m=\omega \bK_m^{-1}\br_m$, where we use a factor $\omega=0.8$ to guarantee stability and ensure that the preconditioned operator remains symmetric positive-definite.

At iteration $k$, we assemble the global preconditioned residual $\mathbf{z}_k$ from the individual nodal corrections $\mathbf{z}_m$. We use this correction, rather than the raw residual, to construct the conjugate search direction $\mathbf{p}_k$:
    \begin{equation}
        \bp_k = \bz_k + \beta_k \bp_{k-1}
    \end{equation} 
and the step length $\alpha_k$:
    \begin{equation}
        \alpha_k = \langle \br_k, \bz_k \rangle / \langle \bp_k, \bK \bp_k \rangle,
    \end{equation}
where the conjugacy coefficient $\beta_k = \langle \br_k, \bz_k \rangle / \langle \br_{k-1}, \bz_{k-1} \rangle$, with $\beta_0 = 0$ at the start. {We use this block-Jacobi preconditioner to reduce the condition number of the system and accelerate convergence. We also compute the Cholesky factorization of $\bK_m$ and cache it prior to starting the preconditioning iterations to further expedite the inverse calculation.}

\subsection{Surface energy $\gamma$ by MicroFract3D}\label{section:simulation-MicroFract3D}

We describe the graph-based model used to simulate the 3D microscopic crack paths and extract the corresponding surface energy $\gamma$, starting with the mathematical formulation of the graph-theoretic energy minimization problem in Section~\ref{section:graph-cut-theory}. We then present our two-phase formulation incorporating metallic grains and spherical voids in Section~\ref{section:surface-energy-two-phase}. We describe the crack path simulation process and the method for extracting the physical surface energy in Section~\ref{section:crack-path-simulation-surface-energy}.

\subsubsection{Fatigue crack path prediction using graph theoretic model}\label{section:graph-cut-theory}

We predict the 3D microscopic crack path in a polycrystalline metallic part under cyclic loading using a graph-based energy minimization framework~\cite{srivastava2021graph}. Rather than simulating incremental crack growth, the algorithm determines the globally optimal crack surface in a single optimization step.
    
To define this optimal surface, the total energy $E_L$ of a partition $L$ is expressed as the sum of a volume integral and a surface integral:
    \begin{equation}
        E_L = \int_{\Omega} g((x,y,z), L(x,y,z)) \, \mathrm{d}V + \int_{\partial \Omega_C} \gamma((x,y,z), \bn) \, \mathrm{d}A,
    \end{equation}
where $\Omega$ is the material volume, $\partial \Omega_C$ represents the fracture surface, $L(x,y,z) \in \{+1, -1\}$ is the spatial label function, $g((x,y,z), L(x,y,z))$ represents the energy density cost of assigning a label $l$ at coordinate $(x,y,z)$, $\gamma((x,y,z), \bn)$ is the local surface energy density, and $\bn$ is the unit normal to the crack surface. The first term, the volume integral, is the data cost associated with the MicroFract3D algorithm, which we do not include in the crack-path surface energy when calculating the fatigue initiation life. The second term represents the pairwise surface energy of interest.

To solve this continuous minimization problem numerically, we discretize the domain by constructing a dual graph $\mathcal{G} = (V, E)$ on a 3D tetrahedral finite element mesh. Each element represents a vertex $v_i \in V$, and each shared interior face between two adjacent elements represents an undirected edge $(v_i, v_j) \in E$. We define a binary partition labeling $L = \{l_1, \dots, l_{N_V}\}$, where $l_i \in \{+1, -1\}$ represents whether element $i$ lies above ($+1$) or below ($-1$) the crack surface. The crack interface itself is represented by the set of element faces that separate vertices with opposite labels $l_i \neq l_j$.

Based on this graph representation, the discrete energy minimization problem is modeled using an Ising-like Hamiltonian on the graph~$\mathcal{G}$, which approximates the continuous energy $E_L$:
    \begin{equation}
        H(L) = \sum_{v_i \in V} g_i(l_i) \, V_i + \sum_{(v_i, v_j) \in E} \gamma_{ij} \, A_{ij} \Bigl(1 - \delta_{l_i, l_j}\Bigr),
        \label{eq:ising-hamiltonian}
    \end{equation}
where $V_i$ is the volume of element $i$ representing the vertex $v_i$, $g_i(l_i)$ is the discrete unary data cost, $A_{ij}$ is the area of the shared face between elements $i$ and $j$, $\gamma_{ij}$ is the interface energy density, and $\delta_{x,y}$ is the Kronecker delta. The volume summation term approximates the continuous volume integral $\int_{\Omega} g \, \mathrm{d}V$, while the second term ensures that the surface energy is added only across the elements that form the crack interface $l_i \neq l_j$. 

\subsubsection{Two-phase formulation}\label{section:surface-energy-two-phase}

We extend the graph-based model to a two-phase microstructure consisting of metallic grains and spherical voids. Specifically, we assign a unique label to the voids and treat them as a homogeneous phase in the simulation.

Within this two-phase model, the local surface energy density $\gamma_{ij}$ is determined by the phase labels of the adjacent elements. First, if elements $i$ and $j$ belong to the same grain, the transgranular cleavage surface energy depends on the face normal $\bn_{ij}$ rotated into the grain's local crystallographic frame:
    \begin{equation}
        \gamma_{ij} = \gamma_{\mathrm{trans}}(\bn_{ij}).
    \end{equation}
We model SS316L using face-centered cubic grains and model cubic anisotropy based on~\cite{srivastava2021graph}, where elevation $\theta$ and azimuth $\phi$ in the local crystal frame are used to model transgranular energy. We use a Cartesian reference:
    \begin{equation}
        n_1 = \cos\theta\cos\phi, \quad n_2 = \cos\theta\sin\phi, \quad n_3 = \sin\theta,
        \label{eq:spherical-to-cartesian}
    \end{equation}
and thus model the normalized transgranular surface energy of the crack using the anisotropy parameter~$\delta$:
    \begin{equation}
        \gamma_{\mathrm{trans}}(\bn) = 1 + \delta - \delta \left( n_1^4 + n_2^4 + n_3^4 - 2(n_1^2 n_2^2 + n_2^2 n_3^2 + n_3^2 n_1^2) \right),
        \label{eq:gamma-trans}
    \end{equation}
where $\bn = (n_1, n_2, n_3)$ is the unit normal to the face rotated to the crystal frame. We detail the value of $\delta$ used in simulations in Section~\ref{section:num-parameters}. Since $\bn$ is a unit normal satisfying $n_1^2 + n_2^2 + n_3^2 = 1$, squaring both sides and substituting into eq.~\eqref{eq:gamma-trans}, it simplifies to
    \begin{equation}
        \gamma_{\mathrm{trans}}(\bn) = 1 + 4\delta(n_1^2 n_2^2 + n_2^2 n_3^2 + n_3^2 n_1^2).
    \end{equation}
Under this formulation, the energy evaluates to a minimum of $1.0$ on the easiest $\{001\}$ cleavage planes because $\bn=(0,0,1)^{\top}$. The normalized energy $\gamma_{\mathrm{trans}}(\bn)$ increases on other planes, directing the crack along preferred crystallographic paths.

Second, if elements $i$ and $j$ belong to two different metallic grains, the intergranular crack resistance is governed by the misorientation angle $\Delta\theta_{ij}$:
    \begin{equation}
        \gamma_{ij} = E_{GG} \left( 1 + 0.1 \, \frac{|\Delta\theta_{ij}|}{\pi} \right),
        \label{eq:gamma-inter}
    \end{equation}
where $E_{GG}$ is the baseline grain-grain boundary energy parameter, and $\Delta\theta_{ij}$ is computed from the grain Rodrigues orientations. Setting $E_{GG} = 1.2$ makes grain boundaries 20\% more resistant to cracking than the easiest transgranular cleavage plane, which models the ductile nature of 316L stainless steel and prevents premature intergranular failure. 

Third, if one element is a metallic grain and the other is a void, the interface energy is constant:
    \begin{equation}
        \gamma_{ij} = E_{GP},
        \label{eq:gamma-gp}
    \end{equation}
where $E_{GP}$ is the grain-void surface energy. We use a low value ($E_{GP} \ll 1.0$), which reflects the physical reality that voids represent pre-existing empty space requiring very little energy to crack, thereby acting as strong attractors that guide the crack path. Finally, if both elements lie within a void, we set $\gamma_{ij} = 0$ since no energy is required to propagate a crack through empty space. 

\subsubsection{Crack path simulation and surface energy extraction}\label{section:crack-path-simulation-surface-energy}

To implement the graph model numerically, a Delaunay triangulation is used to generate a random 3D tetrahedral mesh. To resolve the complex geometry of spherical voids, we sample boundary voxels at the interface of the void and matrix phases using 6-neighbor voxel connectivity. We inject these coordinate locations as mandatory vertices during triangulation to ensure that the mesh naturally refines near void boundaries to capture local geometric features. Each tetrahedral element is mapped to its corresponding grain or void phase via a centroid lookup.

Once the mesh is constructed, we simulate Mode-I tensile loading by locking the labels at the boundaries. We force the top $5\%$ of the elements to the label $l_i = +1$ and the bottom $5\%$ to $l_i = -1$. The left boundary at the mid-plane is split to seed the initiation of the crack. We introduce a weak parabolic potential well centered at $z = 0.5$ to serve as a macroscopic~constraint:
    \begin{equation}
        g_i(l_i) := U_0 \, \mathcal{H}\left( -l_i (z_i - 0.5) \right) (z_i - 0.5)^2,
        \label{eq:unary-well}
    \end{equation}
where $U_0$ is the data factor and $\mathcal{H}$ is the Heaviside step function. This term acts as a soft constraint that keeps the crack path near the mid-plane while letting the microstructural grains and voids dictate the local details of the path.

To find the partition that minimizes the discrete Hamiltonian~\eqref{eq:ising-hamiltonian}, we employ the graph-cut optimization library~\cite{delong2012fast}. Because this library requires integer weights, we scale both the unary and pairwise matrices by a large factor and cast them to 32-bit integers before solving. We minimize the Hamiltonian~\eqref{eq:ising-hamiltonian} globally using the alpha-expansion algorithm. After the solver converges, we extract the physical crack surface and compute the final surface energy $\gamma$ by summing only the physical pairwise edge weights crossing the partition interface:
    \begin{equation}
        \gamma = \sum_{(v_i, v_j) \in E, \, l_i \neq l_j} \gamma_{ij} \, A_{ij}.
        \label{eq:extracted-gamma}
    \end{equation}
We explicitly exclude the unary data cost from the final surface energy calculation so that the output represents a pure crack-resistance score of the microstructural path. 
    
Note that since the simulation computes a normalized dimensionless surface energy score, where the baseline transgranular cleavage energy is 1.0, we multiply this normalized score by the base surface energy $\gamma_0$ and the physical voxel area. We list the values we use for surface energy calculations, including $E_{GG}$, $E_{GP}$, $U_0$, $\gamma_0$ and the large integer scaling factor, in Section~\ref{section:num-parameters}.
    
\subsection{Plastic work $\omega^p$ by crystal plastic finite element analysis}\label{section:simulation-CPFE}

We perform crystal plasticity finite element analysis to compute the local plastic work and extract the FIP $\omega^p$. {Following the PRISMS-Plasticity software~\cite{yaghoobi2019prisms},} we first describe the crystal plasticity constitutive model in Section~\ref{section:CPFE-constitutive} and its numerical incremental formulation with the equilibrium solver in Section~\ref{section:CPFE-incre-equi}. {Based on the output of PRISMS-Plasticity simulation,} we then discuss the FIP extraction using the accumulated plastic work in Section~\ref{section:CPFE-FIP}. 

\subsubsection{Constitutive model}\label{section:CPFE-constitutive}

The PRISMS-Plasticity software implements a rate-independent single crystal plasticity theory~\cite{anand1996computational} with support for multiphase materials. Under this framework, we adopt a two-phase material model to represent the metal grains with embedded voids. We use a face-centered cubic crystal structure with 12 slip systems for the metallic phase and model voids as a material with near-zero elastic stiffness. We assign the single-crystal elastic stiffness constants and crystal plasticity parameters based on the metallic material used for printing. For voids, a small fraction of the metal stiffness and slip resistance is adopted so that the void phase remains elastic and does not undergo plastic slip. This avoids numerical instabilities associated with zero-stiffness void elements. %

To describe the deformation kinematics, two independent deformation mechanisms accommodate the applied deformation: the elastic distortion of the crystal lattice and the pure shear induced by the plastic slip. We use a finite deformation framework with a multiplicative decomposition of the deformation gradient tensor
    \begin{equation}
        \bF := \bF^{e}\bF^{p},
    \end{equation}
where $\bF^{e}$ is the elastic distortion and $\bF^{p}$ accounts for the plastic deformation through crystallographic slip. We define the plastic velocity gradient $\bL^{p}$, the macroscopic rate of plastic distortion, as a sum of shear rates on the slip systems:
    \begin{equation}
        \bL^{p} := \sum_{\alpha=1}^{N_s} \dot{\gamma}^{\alpha}\,\bS^{\alpha},
    \end{equation}
where $\dot{\gamma}^{\alpha}$ is the shearing rate on slip system $\alpha$, and $\bS^{\alpha} = \bR\,(\bm^{\alpha}\otimes\bn^{\alpha})\, \bR^{\top}$ is the Schmid tensor rotated to the sample frame and $\bR$ is the crystal orientation rotation matrix. The slip direction and the slip plane normal are $\bm^{\alpha}$ and $\bn^{\alpha}$, respectively.

Next, we write the constitutive response using the second Piola--Kirchhoff stress in the intermediate~configuration
    \begin{equation}
        \bT^{*} = \cC[\bE^{e}] = \cC\left[\frac{1}{2}\left(\bC^{e} - \bI\right)\right],
    \end{equation}
where $\cC$ is the fourth-order anisotropic elastic stiffness tensor and $\bE^{e} = \frac{1}{2}(\bC^{e} - \bI)$ is the elastic Green--Lagrange strain tensor, with $\bC^{e} = (\bF^{e})^{\top}\bF^{e}$.

To model plastic flow, we compute the resolved shear stress on each slip system as
    \begin{equation}
        \tau^{\alpha} = (\bC^{e}\,\bT^{*}) : \bS^{\alpha}.
        \label{eq:resolved-shear}
    \end{equation}
Rate-independent slip activation is thus governed by the yield function on each slip system,
    \begin{equation}
        f^{\alpha} = \lvert \tau^{\alpha} - W^{\alpha}_{\mathrm{kh}} \rvert - s^{\alpha} \le 0,
    \end{equation}
where $s^{\alpha}$ is the isotropic slip resistance and $W^{\alpha}_{\mathrm{kh}}$ is the kinematic backstress. The isotropic resistance $s^\alpha$ evolves with latent hardening parameters $q_{\alpha\beta}$, while the kinematic backstress $W^\alpha_{\mathrm{kh}}$ captures the cyclic Bauschinger effect, as detailed in~\cite{yaghoobi2019prisms}. Finally, we recover the spatial Cauchy stress $\bsigma$ and the work-conjugate first Piola--Kirchhoff stress $\bP$ via
    \begin{equation}
        \bsigma = \frac{1}{\det(\bF^{e})}\,\bF^{e}\,\bT^{*}\,(\bF^{e})^{T}, \quad \text{and} \quad \bP = \det(\bF)\,\bsigma\,\bF^{-T}.
    \end{equation}

\subsubsection{Incremental formulation and equilibrium solver}\label{section:CPFE-incre-equi}

To solve the mechanical problem under cyclic loading, we apply displacement-controlled, fully reversed cyclic loading in the $z$-direction and discretize the loading history into increments. In each increment, we compute a trial elastic deformation gradient $\bF^{e}(t + \Delta t)^{\mathrm{tr}}$ assuming purely elastic deformation from time $t$ to $t + \Delta t$:
\begin{equation}
    \bF^{e}(t + \Delta t)^{\mathrm{tr}} = \bF(t + \Delta t)\,\bF^{p}(t)^{-1},
\end{equation}
where $\bF(t + \Delta t)$ is the total deformation gradient at time $t + \Delta t$, and $\bF^{p}(t)$ is the plastic deformation gradient at the start of the increment. 
    
Based on this trial elastic state, we calculate the trial resolved shear stress $\tau^{\alpha}(t + \Delta t)^{\mathrm{tr}}$ evaluated on slip system $\alpha$ at time $t + \Delta t$~as
    \begin{equation}
        \tau^{\alpha}(t + \Delta t)^{\mathrm{tr}} = \left(\bC^{e}(t + \Delta t)^{\mathrm{tr}}\,\bT^{*}(t + \Delta t)^{\mathrm{tr}}\right) : \bS^{\alpha},
    \end{equation}
where $\bC^{e}(t + \Delta t)^{\mathrm{tr}} = (\bF^{e}(t + \Delta t)^{\mathrm{tr}})^{T}\bF^{e}(t + \Delta t)^{\mathrm{tr}}$ is the trial right Cauchy--Green elastic deformation tensor, and $\bT^{*}(t + \Delta t)^{\mathrm{tr}} = \cC\left[\frac{1}{2}\left(\bC^{e}(t + \Delta t)^{\mathrm{tr}} - \bI\right)\right]$ is the trial second Piola--Kirchhoff stress. We define a trial yield function $f^\alpha(t + \Delta t)^{\mathrm{tr}}$ for slip system $\alpha$ at $t+\Delta t$ as:
    \begin{equation}
        f^\alpha(t + \Delta t)^{\mathrm{tr}} = \lvert\tau^{\alpha}(t + \Delta t)^{\mathrm{tr}} - W^{\alpha}_{\mathrm{kh}}(t)\rvert - s^{\alpha}(t),
    \end{equation}
where $W^{\alpha}_{\mathrm{kh}}(t)$ is the kinematic backstress and $s^{\alpha}(t)$ is the isotropic slip resistance at the start of the increment. If the trial yield function $f^\alpha(t + \Delta t)^{\mathrm{tr}} > 0$, the slip system $\alpha$ is active. We denote the set of potentially active systems by $\mathcal{A}$. To determine the plastic shear increments on $\mathcal{A}$, we solve
    \begin{equation}
        \sum_{\beta \in \mathcal{A}} A^{\alpha\beta}\,\Delta\gamma^{\beta} = f^{\alpha}(t + \Delta t)^{\mathrm{tr}} \quad \forall \alpha \in \mathcal{A},
    \end{equation}
where the index $\beta \in \mathcal{A}$ runs over all active systems that couple with system $\alpha$, $\Delta\gamma^\beta$ is the unknown plastic shear increment on slip system $\beta$, and $A^{\alpha\beta}$ represents the coupling between the active systems. Note that the slip-system interaction is due to elastic stiffness and hardening and is captured by the matrix $A^{\alpha\beta}$.

Because active slip systems can be redundant or become inactive during the loading increment, we iteratively adjust the active set $\mathcal{A}$ by removing any system where the computed increment is non-positive until all remaining systems in $\mathcal{A}$ satisfy $\Delta\gamma^\alpha > 0$. Once the active set $\mathcal{A}$ and the plastic shear increments $\Delta\gamma^\alpha$ have converged, we update the plastic deformation gradient $\bF^{p}(t + \Delta t)$ at the end of the increment using a matrix exponential map:
    \begin{equation}
        \bF^{p}(t + \Delta t) = \exp\!\left(\Delta\bL^{p}\right)\,\bF^{p}(t),
    \end{equation}
where $\Delta\bL^{p} = \sum_{\alpha \in \mathcal{A}} \operatorname{sign}\!\left[\tau^{\alpha}(t + \Delta t)^{\mathrm{tr}} - W^{\alpha}_{\mathrm{kh}}(t)\right]\,\Delta\gamma^{\alpha}\,\bS^{\alpha}$. Recall that $\tau^{\alpha}(t + \Delta t)^{\mathrm{tr}}$ is the trial shear stress evaluated at time $\tau$, $W^{\alpha}_{\mathrm{kh}}(t)$ is the kinematic backstress, and $\bS^{\alpha}$ is the Schmid tensor rotated to the sample frame. Using the updated $\bF^p(t + \Delta t)$, we recompute the elastic deformation gradient $\bF^{e}(t + \Delta t) = \bF(t + \Delta t)\bF^{p}(t + \Delta t)^{-1}$, update the stress state, and advance the internal hardening variables to time $t + \Delta t$. 

With the local stress update completed at each quadrature point, the global boundary value problem enforces mechanical equilibrium across the domain. The weak form of equilibrium is written as:
    \begin{equation}
        G(\bu, \delta\bu) \equiv \int_{\beta_0} \bP : \nabla_0 \delta\bu \, \mathrm{d}V - \int_{\partial\beta_0} \blambda \cdot \delta\bu \, \mathrm{d}A = 0
    \end{equation}
for all kinematically admissible virtual displacements $\delta\bu$, where $\bP$ is the first Piola--Kirchhoff stress, $\blambda$ is the boundary traction, and $\beta_0$ is the reference configuration. We solve this global nonlinear equilibrium equation $G(\bu, \delta\bu) = 0$ using the Newton--Raphson method. 

\subsubsection{Plastic work-based fatigue indicator parameter}\label{section:CPFE-FIP}

From the converged stress and deformation fields, the crystal plasticity finite element solver computes the cumulative work density $w_e(t)$ at each mesh element $e$ by integrating the stress over the loading history using a trapezoidal rule:
    \begin{equation}
        w_e(t) = \sum_{n=0}^{N(t)} \sum_{q} \frac{1}{2}\left(\bP^{n+1} + \bP^{n}\right) : \Delta\bF^n\, J_q\,w_q,
    \end{equation}
where $N(t)$ is the number of time increments up to time $t$, $q$ is the quadrature point index in the element $e$, $J_q$ and $w_q$ are the Jacobian determinant and quadrature weight at point $q$, respectively, $\bP^n$ is the first Piola-Kirchhoff stress at time instance $t_{n}$, and $\Delta\bF^n = \bF^{n+1} - \bF^{n}$ is the increment of the deformation gradient during increment $n$.
    
The second loading cycle is a stabilized cycle, and {we define the local FIP $\omega_{p,e}$ for element $e$ as the net work density accumulated during the second loading cycle:
    \begin{equation}
        \omega_{p,e} := w_e(t_{\mathrm{end,\,cycle\,2}}) - w_e(t_{\mathrm{end,\,cycle\,1}}).
    \end{equation}
Because the elastic strain energy is state-dependent and returns to the same value at the end of every loading cycle, subtracting the cumulative work density between the endpoints of cycle 1 and cycle 2 exactly isolates the net dissipated plastic work density over the cycle.}
    
{To identify the critical site for fatigue initiation, we exclude all elements belonging to the void phase, and compute the peak FIP over the metallic phase domain $\mathcal{M}$ as:
        \begin{equation}
            \mathrm{FIP}_{\mathrm{peak}} := \max_{e \in \mathcal{M}}\,\omega_{p,e},
        \end{equation}
and use it as the FIP $\omega^p$ to characterize the most critical local plastic deformation hotspot, which correlates with microstructural fatigue crack initiation.}

\section{Uncertainty in fatigue initiation life}\label{section:uncertainty-initiation-life}

The fatigue initiation life of an additively manufactured metallic part is inherently uncertain. This uncertainty stems from the complex microstructure produced during the printing process and the challenges of in situ microstructure characterization. During the printing process, a laser beam selectively melts and fuses the powder particles according to the desired part model data. The rapid heating of the upper surface combined with the slow heat conduction of the underlying layers creates steep temperature gradients. These gradients generate convection currents that affect material flow, rendering the resulting microstructure of metallic grains and nonmetallic voids highly variable. In addition, during the printing process, it is not feasible to accurately measure the geometric properties of the microstructure features, such as grain orientations, lengths, diameters, and void shapes, sizes, and spatial distributions. Consequently, such uncertainties propagate to the fatigue-affecting quantities and ultimately to the fatigue initiation life.

We develop an uncertainty propagation framework to determine the probability distribution of the fatigue initiation life $N$. We first outline the statistical distribution fitting process for the individual fatigue-affecting quantities $G$, $\gamma$, and $\omega^p$ using the block maxima methodology~\cite{romano2017qualification} in Section~\ref{section:pdf-fitting}. We then derive the analytical, closed-form probability density function (PDF) for the fatigue initiation life $N$ under the assumption that $G$, $\gamma$, and $\omega^p$ are correlated normal variables in Section~\ref{section:normal-log-normal}.

\subsection{Distribution fitting}\label{section:pdf-fitting}

We first detail the block maxima approach for extracting representative extreme values from all blocks of a specific region. We then present the generalized extreme value distribution and compare it with the normal and lognormal distributions using model selection criteria.

Rather than fitting distributions to all local data points across the entire microstructure, we adopt the block maxima methodology to capture the critical crack-initiation-controlling values. As shown in Figure~\ref{fig:block-maxima}, within a 3D-printed specimen, we first select a cube that contains a representative microstructure, including both metallic grains and spherical voids, and subdivide it into smaller subcubes for extreme value analysis.

    \begin{figure}
        \centering
        \input{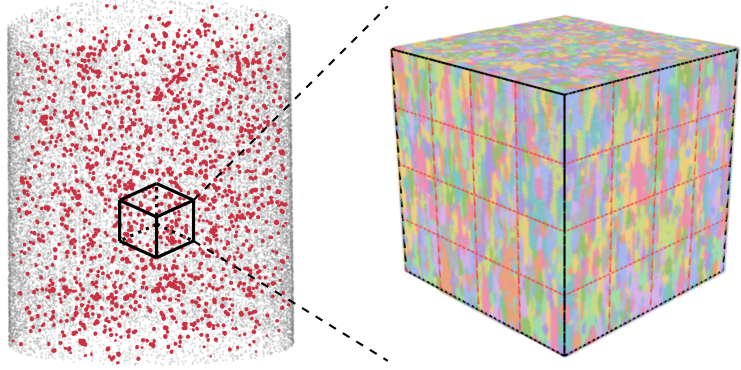}
        \caption{Schematic of the block maxima subcube extraction. The red dots in the part (left) are defects obtained from micro-computed tomography scans. The black wireframe cube in the printed part indicates the local region modeled as the global cube shown on the right. We partition it into subcubes for extreme value analysis. On the surface of the global cube, we show grains with different colors and the partition into 64 subcubes with red lines.}
        \label{fig:block-maxima}
    \end{figure}
    
Next, we run the three simulations discussed above and extract one representative value for each of the three quantities from each small block. For the elastic energy release rate $G$, we select the maximum $G$ value within each subcube. This represents the most critical pre-existing flaw in the subcube. For the surface energy $\gamma$, we select the minimum surface energy value obtained from the graph-cut crack path optimization. This represents the most probable crack path. For the plastic work $\omega^{p}$, we select the maximum peak accumulated plastic work density FIP within each subcube. This represents the plastic deformation hotspot.
        
The block maxima are typically fitted to a generalized extreme value distribution (GEV), an asymptotic limit distribution given by:
        \begin{equation}
            F(x; \mu, \sigma, \xi) = \exp \left( -\left[ 1 + \xi \left( \frac{x - \mu}{\sigma} \right) \right]^{-1/\xi} \right),
            \label{eq:GEV-pdf}
        \end{equation}    
where $\mu \in \real$ is the location parameter, $\sigma > 0$ is the scale parameter, and $\xi \in \real$ is the shape parameter. {The shape parameter determines the tail behavior: $\xi > 0$ corresponds to the heavy-tailed Fr{\'e}chet distribution, $\xi < 0$ to the upper-bounded Weibull distribution, and the limit $\xi \to 0$ to the light-tailed Gumbel distribution.}

{We fit GEV, normal, and lognormal distributions, to the block maxima from the set of subcubes and perform} maximum likelihood estimation for each distribution. {To compare the fitted distribution models, we} compute the Akaike Information Criterion (AIC) and Bayesian Information Criterion (BIC):
    \begin{equation}
        \text{AIC} = 2k - 2\ln(\hat{L}),
        \label{eq:AIC}
    \end{equation}
    \begin{equation}
        \text{BIC} = k\ln(n) - 2\ln(\hat{L}),
        \label{eq:BIC}
    \end{equation}
where $k$ is the number of parameters needed to define a distribution, $n$ is the number of data points, and $\hat{L}$ is the maximum likelihood. We use $n$ data points to estimate all three distributions. For a generalized extreme value distribution, there are three parameters: the location parameter $\mu$, the scale parameter $\sigma$, and the shape parameter $\xi$ as mentioned in eq.~\eqref{eq:GEV-pdf}; thus, the number of parameters is $k=3$. For a normal distribution with mean $\mu_N$ and standard deviation $\sigma_N$, the number of parameters is $k$ = 2. For a lognormal distribution, we fit $k=3$ parameters, namely the location parameter $\theta$, the shape parameter $s$, and the scale parameter $m$.

\subsection{Fatigue initiation life distribution when $\gamma$, $G$, $\omega^p$ are jointly normal}\label{section:normal-log-normal}

We derive a closed-form expression for the PDF of the fatigue initiation life $N$ as expressed by eq.~\eqref{eq:fatigue_life} {assuming that the three fatigue-affecting quantities are correlated and normally distributed random variables.} We analyze the properties of the numerator and address the effects of correlations among the three quantities. {We make this correlation assumption because all three quantities are obtained by simulation results
using the same microstructure and void inputs. In addition, as demonstrated later in the numerical example, a normal distribution model describes the simulation data well in most~cases.}  

In Corollary~\ref{cor:PDF-N-b}, we then extend the PDF to the case where the stored plastic work does not linearly increase with cycles as described in eq.~\eqref{eq:omega_p_N_W} and the fatigue law follows a power-law relation \cite{heinrich2017method}. Throughout the numerical results in Section~\ref{sec: num}, we adopt the power-law relation in Corollary~\ref{cor:PDF-N-b} with $b=0.95$ for all five cases. 

\begin{proposition}\label{pro:PDF-N}
    For the fatigue initiation life $N$ from eq.~\eqref{eq:fatigue_life}, where the surface energy $\gamma \sim \mathcal{N}(\mu_1, \sigma_1^2)$, the elastic energy release rate $G \sim \mathcal{N}(\mu_2, \sigma_2^2)$, and the plastic work, $\omega^p \sim \mathcal{N}(\mu_3, \sigma_3^2)$ are jointly normal random variables with covariances $\operatorname{cov}(\gamma, G) = \rho_{12} \sigma_1 \sigma_2$, $\operatorname{cov}(\gamma, \omega^p) = \rho_{13} \sigma_1 \sigma_3$, and $\operatorname{cov}(G, \omega^p) = \rho_{23} \sigma_2 \sigma_3$. The PDF of the fatigue initiation life $N$ is then given by:
    \begin{equation}
        p_N(n) = p_Z\left(n - c\right),
        \label{eq:PDF-N-final}
    \end{equation}
    where
    \begin{equation}
        p_Z(z) = \frac{1}{2\pi \sigma_{x''} \sigma_3 A(z)^2} \exp\left(-\frac{C}{2}\right) \left[ \sqrt{2\pi} \left(\frac{B(z)}{A(z)}\right) \exp\left( \frac{B(z)^2}{2A(z)^2} \right) \operatorname{erf} \left( \frac{B(z)}{\sqrt{2}A(z)} \right) + 2 \right],
        \label{eq:PDF-Z-definition}
    \end{equation}
    and the auxiliary parameters are:
    \begin{equation*}
        A(z) = \sqrt{\frac{z^2}{\sigma_{x''}^2} + \frac{1}{\sigma_3^2}}, \quad
        B(z) = \frac{\mu_{x''}}{\sigma_{x''}^2}z + \frac{\mu_3}{\sigma_3^2}, \quad
        C = \frac{\mu_{x''}^2}{\sigma_{x''}^2} + \frac{\mu_3^2}{\sigma_3^2}, \quad
        c = \rho_{a3}\frac{\sigma_a}{\sigma_3}.
    \end{equation*}
    
    Here, the mean and variance of the decoupled numerator fluctuation $x''$ are:
    \begin{equation*}
        \mu_{x''} = \mu_a - \rho_{a3}\mu_3 \frac{\sigma_a}{\sigma_3}, \quad
        \sigma_{x''}^2 = \sigma_a^2(1-\rho_{a3}^2),
    \end{equation*}
    with the numerator mean and variance given by $\mu_a = 2\mu_1 - \mu_2$ and $\sigma_a^2 = 4\sigma_1^2 + \sigma_2^2 - 4\rho_{12}\sigma_1\sigma_2$, and the numerator-denominator correlation coefficient $\rho_{a3}$ defined as:
    \begin{equation}
        \rho_{a3} = \frac{2\rho_{13}\sigma_1 - \rho_{23}\sigma_2}{\sqrt{4\sigma_1^2 + \sigma_2^2 - 4\rho_{12}\sigma_1\sigma_2}}.
         \label{eq:rho-a-3}
    \end{equation}
\end{proposition}

\begin{proof}
    The derivation of the closed-form PDF of $N$ proceeds as follows: we define the numerator in eq.~\eqref{eq:fatigue_life} as $a := 2\gamma - G$. Since it is a linear combination of the jointly normal random variables $\gamma$ and $G$, it is also normally distributed, $a \sim \mathcal{N}(\mu_a, \sigma_a^2)$, with mean $\mu_a = 2\mu_1 - \mu_2$ and variance $\sigma_a^2 = 4\sigma_1^2 + \sigma_2^2 - 4\rho_{12}\sigma_1\sigma_2$. The correlation coefficient $\rho_{a3}$ between the numerator $a$ and the denominator $\omega^p$ is computed by evaluating the covariance:
        \begin{equation}
            \operatorname{cov}(2\gamma - G, \omega^p) = 2\operatorname{cov}(\gamma, \omega^p) - \operatorname{cov}(G, \omega^p) = 2\rho_{13}\sigma_1\sigma_3 - \rho_{23}\sigma_2\sigma_3.
        \end{equation}
    Dividing by the product of their standard deviations yields eq.~\eqref{eq:rho-a-3}. Since the numerator $a$ and the denominator $\omega^p$ are correlated, we cannot directly apply the ratio distribution for independent normal variables. To decouple this correlation, we express $a$ and $\omega^p$ in terms of their mean values and the zero-mean fluctuations $x \sim \mathcal{N}(0, \sigma_a^2)$ and $y \sim \mathcal{N}(0, \sigma_3^2)$, such that:
        \begin{equation}
            N = \frac{a}{\omega^p} = \frac{x + \mu_a}{y + \mu_3}.
        \end{equation}
    We then construct a new zero-mean normal variable $x'$ that is orthogonal (uncorrelated) to $y$ via:
        \begin{equation}
            x' = x - \rho_{a3}\frac{\sigma_a}{\sigma_3}y,
        \end{equation}
    which has a variance of $\sigma_a^2(1 - \rho_{a3}^2)$ and satisfies $\operatorname{cov}(x', y) = 0$. We substitute $x = x' + \rho_{a3}\frac{\sigma_a}{\sigma_3}y$ back into the expression for $N$:
        \begin{equation}
            N = \frac{x' + \rho_{a3}\frac{\sigma_a}{\sigma_3}y + \mu_a}{y + \mu_3} = \frac{x' + \mu_a - \rho_{a3}\mu_3\frac{\sigma_a}{\sigma_3}}{y + \mu_3} + \rho_{a3}\frac{\sigma_a}{\sigma_3}.
        \end{equation}
    By defining the mean-shifted variable $x'' = x' + \mu_a - \rho_{a3}\mu_3\frac{\sigma_a}{\sigma_3}$, which is normally distributed as $x'' \sim \mathcal{N}(\mu_{x''}, \sigma_{x''}^2)$, we can write:
        \begin{equation}
            N - c = \frac{x''}{\omega^p},
        \end{equation}
    where $c = \rho_{a3}\frac{\sigma_a}{\sigma_3}$ and $\omega^p = y + \mu_3$. Because $\operatorname{cov}(x'', \omega^p) = \operatorname{cov}(x', y) = 0$, $x''$ and $\omega^p$ are uncorrelated normal random variables.
    
    Since $x''$ and $\omega^p$ are uncorrelated, their ratio $Z = x''/\omega^p$ follows the ratio distribution of uncorrelated normal variables. According to~\cite{hinkley1969ratio}, the PDF of the ratio $Z=X/Y$ of two uncorrelated normal random variables $X\sim \mathcal{N}(\mu_x, \sigma_x^2)$ and $Y\sim \mathcal{N}(\mu_y, \sigma_y^2)$ is given by:
    \begin{equation}
        p_Z(z) = \frac{1}{2\pi \sigma_x \sigma_y \tilde{A}(z)^2} \exp\left(-\frac{\tilde{C}}{2}\right) \left[ \sqrt{2\pi} \frac{\tilde{B}(z)}{\tilde{A}(z)} \exp\left( \frac{\tilde{B}(z)^2}{2\tilde{A}(z)^2} \right) \operatorname{erf} \left( \frac{\tilde{B}(z)}{\sqrt{2}\tilde{A}(z)} \right) + 2 \right],
        \label{eq:PDF-Z}
    \end{equation}
    where $\operatorname{erf}(\cdot)$ is the Gaussian error function, and the auxiliary parameters are:
    \begin{equation*}
        \tilde{A}(z) = \sqrt{\frac{1}{\sigma_x^2}z^2+\frac{1}{\sigma_y^2}},\quad 
        \tilde{B}(z) = \frac{\mu_x}{\sigma_x^2}z + \frac{\mu_y}{\sigma_y^2},\quad 
        \tilde{C} = \frac{\mu_x^2}{\sigma_x^2} +\frac{\mu_y^2}{\sigma_y^2}.
    \end{equation*}
    
    We obtain the PDF $p_Z(z)$ defined in eq.~\eqref{eq:PDF-Z-definition} by substituting $\mu_x$ with $\mu_{x''}$, $\sigma_x$ with $\sigma_{x''}$, $\mu_y$ with $\mu_3$, and $\sigma_y$ with $\sigma_3$. Since $N = Z + c$, we then shift the PDF of $Z$ by the constant $c$, yielding the closed-form PDF of the fatigue initiation life $N$ as $p_N(n) = p_Z(n-c)$.
        
\end{proof}

\newtheorem{corollary}{Corollary}

\begin{corollary}\label{cor:PDF-N-b}
    If the fatigue initiation life is calculated under a power-law relation as:
    \begin{equation}
        N^b = \frac{2\gamma - G}{\omega^p},
        \label{eq:N-power-b}
    \end{equation}
    where $b$ is the exponent of the fatigue life. The probability density function of $N$ is then:
    \begin{equation}
        p_N(n) = b n^{b - 1} \cdot p_Z\left(n^b - \rho_{a3}\frac{\sigma_a}{\sigma_3}\right),
    \end{equation}
    where the auxiliary parameters and $p_Z(z)$ are defined in Proposition~\ref{pro:PDF-N}.
\end{corollary}

\begin{proof}
    Using the transformation of random variables, let $Z = N^b - \rho_{a3}{\sigma_a}/{\sigma_3}$. From Proposition~\ref{pro:PDF-N}, $Z = {x''}/{\omega^p}$ follows the uncorrelated normal ratio distribution $p_Z(z)$ defined in eq.~\eqref{eq:PDF-Z-definition}. The probability density function of $N$ is related to that of $Z$ via the Jacobian of the transformation:
    \begin{equation}
        p_{N}(n) = p_{Z}(z(n)) \left| \frac{\mathrm{d}z}{\mathrm{d}n} \right| = p_{Z}(n^b - c) \left| b n^{b-1} \right|.
    \end{equation}
    For $n > 0$, this yields the closed-form PDF of the fatigue initiation life $N$ under the power-law relation defined in eq.~\eqref{eq:N-power-b}.
\end{proof}

\section{Numerical results}\label{sec: num}

{We present the numerical results of the proposed framework, beginning with descriptions of} the experimental microstructure characterizations of the printed parts, including both metallic grains and voids, in Section~\ref{section:num-inputs}. We then detail the design of numerical experiments, including simulation setup and model parameters, in Section~\ref{section:num-setup}. In Section~\ref{section:num-results}, we show the uncertainty propagation from the fatigue-affecting quantities to the fatigue initiation life using the recommended EOS M290 print parameter set, where we present simulation results, analyze the distributions of each fatigue-affecting quantity, and show the resulting closed-form probability density function for the fatigue initiation life. In addition, we design two numerical studies to understand the effect of grain and defect size variations on the distributions of fatigue initiation life in Sections~\ref{section:num-results-eff-grain} and \ref{section:num-results-eff-void}. {We present the framework for predicting the fatigue life with quantified uncertainties for a specific specimen printed using an EOS M290 laser powder bed fusion machine in Figure~\ref{fig:pipeline}.}

All simulations in this study are executed on a workstation equipped with a 28-core Intel Xeon W-3175X CPU @ 3.10 GHz. Table~\ref{table:cpu-times} summarizes the quantities computed, the number of evaluations per case, the total number of evaluations across all five cases, and the average CPU execution times over 64 runs for each of the three physical models in our workflow.

\begin{table}[htbp]
    \centering
    \caption{Computational cost of the three simulation models in the uncertainty propagation workflow. All evaluations are performed on a 28-core Intel Xeon W-3175X CPU @ 3.10 GHz.}
    \label{table:cpu-times}
    {\footnotesize
    \begin{tabular}{l l c c c}
        \textbf{Model} & \textbf{Quantity Computed} & \makecell{\textbf{CPU Time} \\ \textbf{per Evaluation}}  & \makecell{\textbf{Evaluations} \\ \textbf{per Case}} & \makecell{\textbf{Total} \\ \textbf{Evaluations}}\\
        \hline
        ELAS3D & Energy release rate $G$  & 355 s & 64 & 320\\
        MicroFract3D & Surface energy $\gamma$  & 91 s & 64 & 320\\
        PRISMS-Plasticity & Fatigue indicator parameter $\omega^p$ &  17.97 h & 64 & 320 \\
    \end{tabular}
    }
\end{table}

\subsection{Specimen information}\label{section:num-inputs}

We characterize grains and voids from additively manufactured parts made of 316L stainless steel using the laser powder bed fusion. An EBSD texture analysis of the specimens constructs the 3D grain representation and a micro-computed tomography (microCT) scanning characterizes volumetric flaws across different print parameter sets. These experimental results inform our selections when generating microstructures, including metallic grains and voids, for simulations.

Using an EOS M290 laser powder bed fusion machine, cylindrical specimens with 10 mm diameter and 25 mm height are produced by varying the laser power and scan speed according to the parameter sets in the process map in Figure~\ref{fig:set-param} and in Table~\ref{table:print-parameters}. The CAD models and the printed parts are shown in Figure~\ref{fig:printed-samples}. The produced cylinder specimens have a broad range of microstructures, varying grain lengths, aspect ratios, and orientations, and defect sizes.

    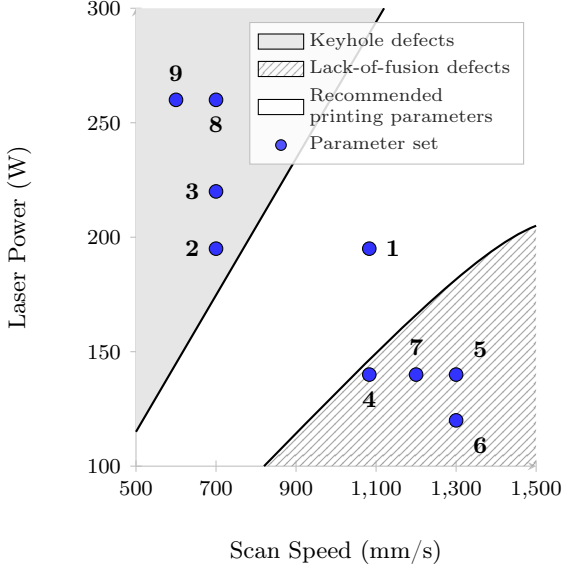
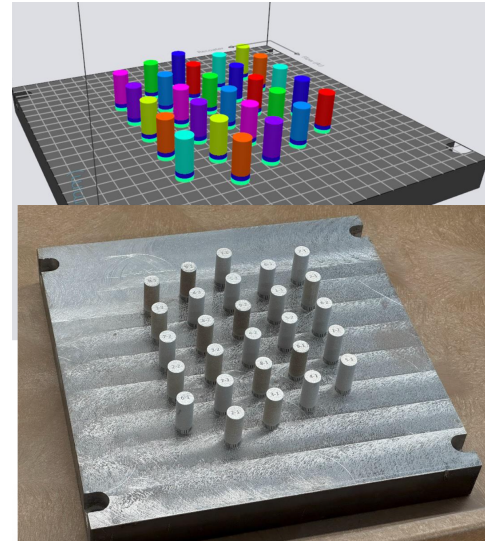
\begin{figure}[htbp]
        \centering
        \begin{subfigure}[t]{0.49\textwidth}
            \centering
            \usetikzlibrary{patterns}%
\begin{tikzpicture}
  \begin{axis}[
    width=0.90\linewidth,
    height=\linewidth,
    axis lines=left,
    xlabel={Scan Speed (mm/s)},
    ylabel={Laser Power (W)},
    xmin=500, xmax=1500,
    ymin=100, ymax=300,
    xtick={500, 700, 900, 1100, 1300, 1500},
    ytick={100, 150, 200, 250, 300},
    enlargelimits=false,
    tick label style={font=\footnotesize},
    label style={font=\small},
    xlabel style={at={(ticklabel cs:0.5)}, anchor=north, yshift=-3mm},
    ylabel style={at={(ticklabel cs:0.5)}, anchor=south, yshift=5mm},
    axis line style={lightgray},
    every tick/.style={lightgray},
    grid=none,
    legend pos=north east,
    legend cell align={left},
    legend style={font=\scriptsize, fill=white, fill opacity=0.85, draw=gray!50, align=left}
  ]
    
    \fill[gray!20] (axis cs:500, 115) -- (axis cs:1120, 300) -- (axis cs:500, 300) -- cycle;
    
    \fill[pattern=north east lines, pattern color=gray!70] (axis cs:820, 100) .. controls (axis cs:1100, 150) and (axis cs:1350, 195) .. (axis cs:1500, 205) 
                  -- (axis cs:1500, 100) -- cycle;
    
    \draw[black, thick] (axis cs:500, 115) -- (axis cs:1120, 300);
    \draw[black, thick] (axis cs:820, 100) .. controls (axis cs:1100, 150) and (axis cs:1350, 195) .. (axis cs:1500, 205);
    
    \node[circle, draw=black, fill=blue!80!white, inner sep=1.8pt, label={[font=\small\bfseries]right:1}] at (axis cs:1083, 195) {};
    \node[circle, draw=black, fill=blue!80!white, inner sep=1.8pt, label={[font=\small\bfseries]left:2}] at (axis cs:700, 195) {};
    \node[circle, draw=black, fill=blue!80!white, inner sep=1.8pt, label={[font=\small\bfseries]left:3}] at (axis cs:700, 220) {};
    \node[circle, draw=black, fill=blue!80!white, inner sep=1.8pt, label={[font=\small\bfseries]below:4}] at (axis cs:1083, 140) {};
    \node[circle, draw=black, fill=blue!80!white, inner sep=1.8pt, label={[font=\small\bfseries, xshift=1pt, yshift=1pt]above right:5}] at (axis cs:1300, 140) {};
    \node[circle, draw=black, fill=blue!80!white, inner sep=1.8pt, label={[font=\small\bfseries, xshift=1pt, yshift=-1pt]below right:6}] at (axis cs:1300, 120) {};
    \node[circle, draw=black, fill=blue!80!white, inner sep=1.8pt, label={[font=\small\bfseries, yshift=1pt]above:7}] at (axis cs:1200, 140) {};
    \node[circle, draw=black, fill=blue!80!white, inner sep=1.8pt, label={[font=\small\bfseries]below:8}] at (axis cs:700, 260) {};
    \node[circle, draw=black, fill=blue!80!white, inner sep=1.8pt, label={[font=\small\bfseries, yshift=1pt]above:9}] at (axis cs:600, 260) {};
    
    \addlegendimage{area legend, fill=gray!20, draw=black}
    \addlegendentry{Keyhole defects}
    
    \addlegendimage{area legend, pattern=north east lines, pattern color=gray!70, draw=black}
    \addlegendentry{Lack-of-fusion defects}

    \addlegendimage{area legend, fill=white, draw=black}
    \addlegendentry{Recommended\\printing parameters}
    
    \addlegendimage{only marks, mark=*, mark options={draw=black, fill=blue!80!white}}
    \addlegendentry{Parameter set}
    
  \end{axis}
\end{tikzpicture}
            \caption{Process map showing the relationship between print parameters and defects. The nine print parameter sets are plotted as numbered points and their exact values are given in Table~\ref{table:print-parameters}.}
            \label{fig:set-param}
        \end{subfigure}
        \hfill
        \begin{subfigure}[t]{0.4\textwidth}
            \centering
            \includegraphics[width=\linewidth]{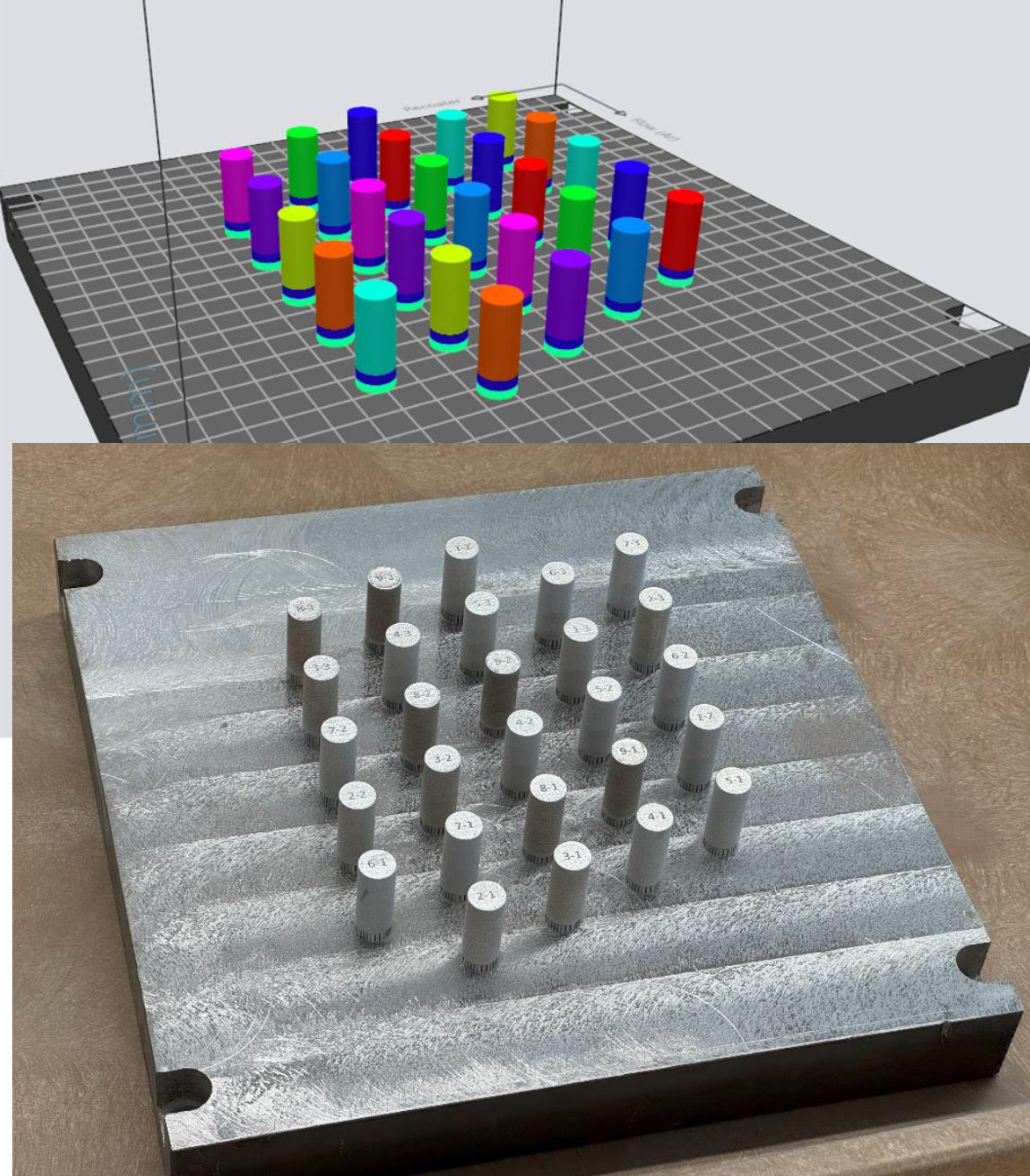}
            \caption{Cylinder specimens are printed according to the different print parameters. For each set, there are multiple samples. Samples printed by recommended and lack-of-fusion parameter sets are marked with blue and cyan colors, respectively.}
            \label{fig:printed-samples}
        \end{subfigure}
        \caption{Overview of print parameter sets and corresponding printed cylinder specimens.}
        \label{fig:combined-process-params}
    \end{figure}
    \begin{table}[htbp]
        \centering
        \caption{Laser power and scan speed for the nine print parameter sets.}
        \label{table:print-parameters}
        \small
        \begin{tabular}{lccccccccc}
            \hline
            \textbf{Parameter set} & \textbf{1} & \textbf{2} & \textbf{3} & \textbf{4} & \textbf{5} & \textbf{6} & \textbf{7} & \textbf{8} & \textbf{9} \\
            \hline
            Laser Power (W) & 195 & 195 & 220 & 140 & 140 & 120 & 140 & 260 & 260 \\
            Scan Speed (mm/s) & 1083 & 700 & 700 & 1083 & 1300 & 1300 & 1200 & 700 & 600 \\
            \hline
        \end{tabular}
    \end{table}
    
EBSD scans are performed for the central part of the specimen using an Oxford Instruments-equipped scanning electron microscope featuring an EBSD camera. Prior to EBSD analysis, the specimens undergo meticulous preparation involving grinding with 600–1200 grit SiC sandpaper, followed by polishing with a 1 $\mu$m diamond suspension and 0.04 $\mu$m colloidal silica to achieve a high-quality surface finish~\cite{andani2024mapping}. Each sample is sectioned at 12.5 mm from the base to perform EBSD scans. A cube subsample is extracted to comprehensively analyze the grain texture along the XY, XZ, and YZ planes at the sample's center and generate a 3D representation of the grain microstructure. Defects in the samples are obtained using high-resolution microCT scans with phase contrast tomography using a Zeiss Xradia microscope. {In Figure~\ref{fig:pipeline}, we show the location of the cube relative to the whole sample and EBSD results along with the distributions of grain diameters and aspect ratios in section (a), as well as defects in two samples printed with different parameters in section (b).}

    \begin{figure}[htbp]
        \input{fig-pipeline-R1}
        \caption{{The fatigue initiation life uncertainty quantification framework. The pipeline starts with EBSD and microCT data of specific prints, and illustrates the workflow that includes simulation and lifetime estimation.
        (a) A cylinder specimen and the portion at the central region of the specimen cut out for EBSD scans. Columnar grain structure on XZ and YZ planes and cell structures at the build layer (XY plane). Shown here is a 1.5 mm square region; we assign unique colors to each grain orientation. Grain diameter and aspect ratio distributions guide grain generations for simulations.
        (b) A set of stacked illustrative microCT slices where defects are red and specimen outlines gray. Specimen (S1) is printed with the recommended EOS M290 parameter Set 1 in Table~\ref{table:print-parameters}. Specimen (S4) is printed with parameter Set 4 in Table~\ref{table:print-parameters} which causes lack-of-fusion defects. Measured total void volume fraction and diameter statistics guide void generations for simulations, including sizes of spherical voids and spacing constraints.
        (c) Generated grains and voids form a shared 1.6-mm cube partitioned into 64 subcubes of 0.4-mm side length for ELAS3D, MicroFract3D, and PRISMS-Plasticity. (d) We run simulations to obtain $G$, $\gamma$, and $\omega^p$. With block maxima processing and PDF fitting, we obtain their distributions and the fatigue initiation life distribution. All distribution curves are representative.}}
        \label{fig:pipeline}
    \end{figure}

\subsection{Design of numerical experiments}\label{section:num-setup}

We describe the numerical setup for the three simulations, namely the modified ELAS3D, MicroFract3D, and PRISMS-Plasticity, used to obtain the fatigue-affecting quantities. Informed by EBSD and microCT measurements, we generate artificial representative grain and defect structures in a 1.6-mm cube and refer to it as the global cube. The grains and defects are generated by sampling grain length, aspect ratio, and void diameter distributions informed by EBSD and microCT results corresponding to a set of print parameters. 

\subsubsection{Microstructure and void generation}

The simulation domain is a global cube with a side length of $1.6$ mm, which corresponds to the center region of a printed columnar specimen, see {section (c) of Figure~\ref{fig:pipeline}. We show grain microstructures on the cube surface; grain lengths and aspect ratios follow lognormal distributions.} We partition it into 64 subcubes, each with a side length of $0.4$ mm. {We also show a specific set of defects; we enlarge the defects for visualization and the spherical defects account for 2\% of the global cube. We enforce the minimum distance between any two defects at 30 $\mu$m. Defect diameters follow a truncated lognormal~distribution. }

{The true geometric characterizations of microstructures corresponding to different print parameter sets are obtained using EBSD and microCT scans. Informed by the measurement data, we address measurement uncertainty and sample-to-sample variations by generating artificial microstructures to conduct uncertainty quantification of fatigue initiation life. For metallic grains, we generate realizations of grain diameters and aspect ratios from lognormal distributions. For spherical voids, we generate the diameters of the spherical voids following a truncated lognormal
distribution. Given a target void volume fraction, we fill the global cube with metallic grains and voids with the same geometric~characteristics.}

We perform parametric studies by varying the grain diameter and void size to investigate their effects on the fatigue initiation life in Section~\ref{section:num-results-eff-grain} and Section~\ref{section:num-results-eff-void}. We summarize all simulation cases in Table~\ref{table:simulation-cases}. The grain and void properties in Cases 1 to 3 are generated synthetically, informed by the specimens produced using the EOS M290 recommended print parameter sets, to study the effects of grain size on the fatigue initiation life. The properties in Case 4 are EBSD and microCT scan data of a specimen printed with a lack-of-fusion print parameter set. Based on the Case 4 void properties, we design Case 5 to study the effects of void sizes on the fatigue initiation life. 

\begin{table}[htbp]
    \centering
    \caption{Summary of simulation cases. In the grain size study (Cases 1 to 3), we vary the grain diameter while keeping the void parameters fixed. In the void size study (Cases 4 and 5), we change the void diameter range while keeping the grain parameters fixed.}
    \label{table:simulation-cases}
    {\footnotesize
    \begin{tabular}{l c c c c c}
        \textbf{Case} & \makecell{\textbf{Grain diameter} \\ \textbf{$\mu$, $\sigma$~($\mu$m)}} & \makecell{\textbf{Aspect ratio} \\ \textbf{$\mu$, $\sigma$}} & \makecell{\textbf{Void volume} \\ \textbf{fraction~(\%)}} & \makecell{\textbf{Void diameter} \\ \textbf{range ($\mu$m)}} & \makecell{\textbf{Void} \\ \textbf{spacing ($\mu$m)}}\\
        \hline
        Case 1 (synthetic) & 60, 18 & 5.0, 1.5 & 2.0 & [50, 100] & 30\\
        Case 2 (synthetic) & 70, 21 & 5.0, 1.5 & 2.0 & [50, 100] & 30\\
        Case 3 (synthetic) & 80, 24 & 5.0, 1.5 & 2.0 & [50, 100] & 30\\
        \hline
        Case 4 (printed) & 25, 15 & 3.3, 1.65 & 1.5 & [10, 80] & 50\\
        Case 5 (synthetic) & 25, 15 & 3.3, 1.65 & 1.5 & [50, 80] & 50\\
    \end{tabular}
    }
\end{table}

\subsubsection{Parameters for the three simulation models}\label{section:num-parameters}

We study the fatigue performance of the printed part under an externally applied strain of 0.3\% amplitude. In all three models, we make sure the boundary conditions satisfy this loading profile. As discussed in Section~\ref{section:two-resolution-scheme}, we use a two-resolution scheme to speed up the ELAS3D simulations. We set the voxel sizes to $0.010$ mm for the coarse mesh and $0.005$ mm for the fine mesh solver in order to capture the local elastic fields. In each subcube, there are $64{,}000$ nodes in the coarse grid and $512{,}000$ nodes in the fine grid. Since there are three translational degrees of freedom at each node, the dimension of $\bK$ in eq.~\eqref{eq:elas3D-linear-system} is $1{,}536{,}000 \times 1{,}536{,}000$ for the fine grid. {The effect of mesh sizes on the convergence of results has been studied in~\cite{jeong2026elas3d}}. The isotropic elastic properties of 316L stainless steel are defined by $E = 194.2$ GPa and $\nu = 0.2934$, yielding a plane-strain modulus of $E' = 212.5$~GPa.
        
In the MicroFract3D simulations, we discretize each subcube using a random 3D tetrahedral mesh of density $60^3$, with coordinates of void-matrix interfaces sampled using 6-neighbor connectivity and injected as mandatory~vertices. Note that each subcube has a unique mesh, which eliminates grid-orientation bias in the graph-cut solver, and the homogenized surface energy converges reliably. We restrict displacements at the top 5\% and bottom 5\% margins and seed crack initiation at the mid-plane along the left boundary. The data factor $U_0 = 0.03$ constrains the crack path near the mid-plane as discussed in Section~\ref{section:graph-cut-theory}, and the pairwise weights are scaled by a factor of $10^6$ for integer-based crack path optimization as discussed in Section~\ref{section:crack-path-simulation-surface-energy}. The physical surface energy is obtained by multiplying the normalized path resistance score by the base surface energy $\gamma_0 = 1.0\text{ J/m}^2$ and the voxel area. We set the intergranular boundary energy to $E_{GG} = 1.2$ and the grain-void interface energy to $E_{GP} = 0.1$ to attract the crack path.
        
In PRISMS-Plasticity simulations, we use a uniform voxel size of 0.01 mm. The single-crystal elastic stiffness values for 316L stainless steel are $C_{11} = 204.6$~GPa, $C_{12} = 137.7$~GPa, and $C_{44} = 126.2$~GPa. Hardening parameters for the 12 FCC slip systems include Voce isotropic parameters ($s_0^\alpha = 95$~MPa, $h_0^\alpha = 350$~MPa, $s_\infty^\alpha = 200$~MPa, $a^\alpha = 2.25$) and Armstrong-Frederick kinematic parameters ($C_1^\alpha = 3000$~MPa, $C_2^\alpha = 100$~MPa). To avoid numerical instabilities, we model voids as elastic elements with a stiffness scaled by $f_{\mathrm{void}} = 10^{-4}$ relative to that of the metallic phase. We set fully reversed cyclic displacement control along the $z$-direction, with a displacement amplitude $\Delta u_z = -0.0012$ mm, a nominal strain amplitude $\varepsilon_a \approx 0.3\%$, and a quarter-cycle period $T_{1/4} = 1.0$ s, for two loading cycles simulated over 800 increments.

\subsection{Uncertainty in fatigue initiation life for recommended print parameters}\label{section:num-results}

We present the uncertainty propagation workflow using the Case 1 microstructure, i.e., using the EOS M290 recommended print parameter set. Using the outputs of the 64 subcube simulations, we first analyze the statistical distributions of the fatigue-affecting quantities and tabulate the parameters of the fitted generalized extreme value distribution, normal distribution, and lognormal distribution for the block extremes of $G$, $\gamma$, and $\omega^p$. We then evaluate the resulting theoretical probability density function of the fatigue initiation life $N$, see eq.~\eqref{eq:PDF-N-final}, under the assumed correlations. 

Using the ELAS3D solver, we perform linear elastic simulations on the 64 subcubes under uniaxial tension. For each subcube, we extract the stress field to calculate the stress intensity factors $K_I$ and the corresponding elastic energy release rate $G = K_I^2 / E'$ for all embedded voids, following the procedure in Section~\ref{section:simulation-LEFM-G}. The graph-theoretic MicroFract3D solver is used to compute the globally optimal 3D fracture surface and determine the surface energy along the resulting crack path as described in Section~\ref{section:simulation-MicroFract3D}. To capture plastic behavior, we run crystal plasticity finite element simulations using PRISMS-Plasticity. Over two loading cycles, the solver logs the local accumulated plastic work density at each increment, from which we compute the plastic work in the second full cycle and extract the peak fatigue indicator parameter $\omega^p$ as described in Section~\ref{section:simulation-CPFE}. 

Note that the PRISMS-Plasticity solver outputs the plastic work density in each voxel. The units of the output are MPa. We must multiply this value by 10 $\mu$m, the voxel thickness in the crystal plasticity finite element mesh. This ensures dimensional consistency when calculating the crack initiation life according to eq.~\eqref{eq:fatigue_life}.
        
\subsubsection{Fitted distributions of $G$, $\gamma$, $\omega^p$}

For the fatigue-affecting quantities $G$, $\gamma$, and $\omega^p$, we fit the generalized extreme value distribution, normal distribution, and lognormal distribution using block extremes obtained from the 64 subcubes. We then find the best probability distribution for each quantity using the Akaike Information Criterion and Bayesian Information Criterion as defined in eq.~\eqref{eq:AIC} and eq.~\ref{eq:BIC}, where smaller values are better. The estimated generalized extreme value distribution parameters are shape $\xi$, scale $\sigma$, and location $\mu$. The normal fit parameters are mean $\mu$ and variance $\sigma^2$. The lognormal fit parameters are shape $s$, location $\theta$, and scale $m$.

We calculate the elastic energy release rate $G$, surface energy $\gamma$, and fatigue indicator parameter $\omega^p$ values as described in Sections~\ref{section:simulation-LEFM-G}, \ref{section:simulation-MicroFract3D}, and \ref{section:simulation-CPFE}, respectively. We tabulate the peak value in each subcube and fit a generalized extreme value distribution, a normal distribution, and a lognormal distribution to the 64 block extremes. We compare the fitting quality using AIC and BIC in Table~\ref{table:fit-comparison-Case1} and visualize the fits in Figure~\ref{fig:CPFE-three-dists}. As listed in Table~\ref{table:fit-comparison-Case1}, for $G$ and $\gamma$, the normal distributions have the smallest BICs; for $\omega^p$, the lognormal distribution achieves the best BIC, but the normal distribution's BIC is not significantly different. Therefore, the normal distribution is selected for all three quantities for consistency with the analytical PDF derivation in Section~\ref{section:normal-log-normal}.

\begin{table}[h]
    \centering
    \caption{Case 1: comparison of probability distribution fits for $G$, $\gamma$, and $\omega^p$. Bold values indicate the best (lowest) BIC.}
    \label{table:fit-comparison-Case1}
    {\footnotesize
    \begin{tabular}{l l l c c}
        \textbf{Quantity} & \textbf{Distribution} & \textbf{Parameters} & \textbf{AIC} & \textbf{BIC} \\
        \hline
        \multirow{3}{*}{$G$}
        & GEV ($\xi=-0.264$) & $\mu=2.18$, $\sigma=0.27$ & 24.21 & 30.68 \\
        & Normal & $\mu_N=2.28$, $\sigma_N=0.28$ & 24.43 & \textbf{28.75} \\
        & Lognormal & $s=0.14$, $\theta=0$, $m=2.04$ & 25.55 & 32.03 \\
        \hline
        \multirow{3}{*}{$\gamma$}
        & GEV ($\xi=-0.383$) & $\mu=5.21$, $\sigma=0.096$ & $-120.66$ & $-114.19$ \\
        & Normal & $\mu_N=5.24$, $\sigma_N=0.091$ & $-120.71$ & $\mathbf{-116.39}$ \\
        & Lognormal & $s=5.58{\times}10^{-6}$, $\theta=1.64{\times}10^{4}$, $m=1.64{\times}10^{4}$ & $-118.71$ & $-112.23$ \\
        \hline
        \multirow{3}{*}{$\omega^p$}
        & GEV ($\xi=0.778$) & $\mu=4.16{\times}10^{-4}$, $\sigma=1.05{\times}10^{-4}$ & $-996.50$ & $-990.02$ \\
        & Normal & $\mu_N=4.48{\times}10^{-4}$, $\sigma_N=6.47{\times}10^{-5}$ & $-1049.03$ & ${-1044.71}$ \\
        & Lognormal & $s=0.23$, $\theta=1.73{\times}10^{-4}$, $m=2.67{\times}10^{-4}$ & $-1053.05$ & $\mathbf{-1046.57}$ \\
    \end{tabular}
    }
\end{table}

\begin{figure}
    \centering
    \input{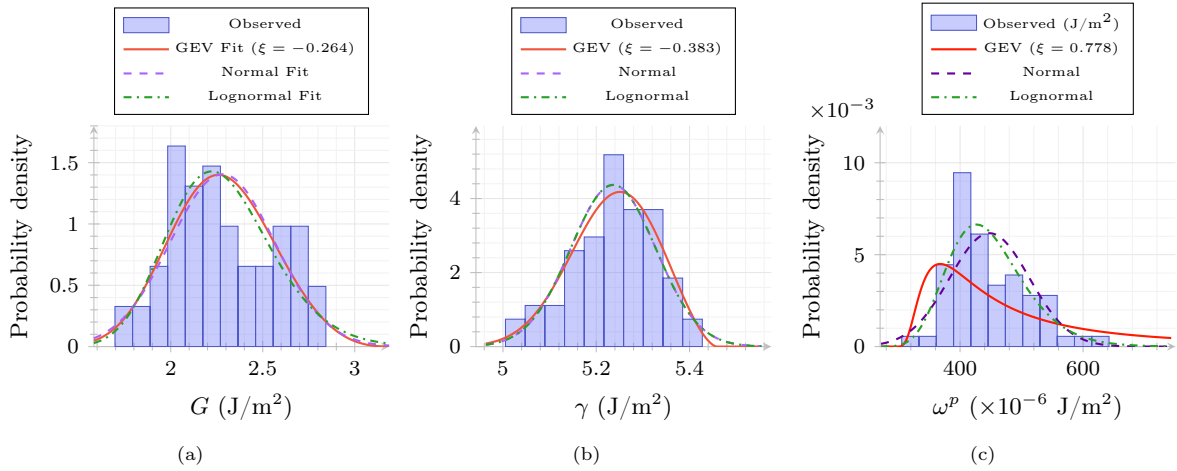}
    \caption{Case 1: block-maxima distribution fits for the three fatigue-affecting quantities. (a)~The peak elastic energy release rate $G$ values in all 64 subcubes; the normal distribution (BIC\,=\,28.75) is preferred over GEV-Weibull ($\xi=-0.264$). (b)~The surface energy $\gamma$ along the minimum-energy crack path; the normal distribution (BIC\,=\,$-116.39$) is preferred over GEV-Weibull ($\xi=-0.383$). (c)~The peak FIP $\omega^p$ values; the lognormal distribution provides the best BIC, but the normal distribution (BIC\,=\,$-1044.71$) is selected for consistency with the analytical PDF derivation.}
    \label{fig:CPFE-three-dists}
\end{figure}

\subsubsection{Theoretical probability density function of fatigue initiation life $N$}

Since all three simulations use the same set of microstructures, including metallic grains and voids, the correlations among the fatigue-affecting quantities should be high, so we choose the following correlation coefficients: $\rho_{G,\gamma} = 0.8$, $\rho_{G,\omega^p} = 0.7$, and $\rho_{\gamma,\omega^p} = 0.6$.
    
We obtain the probability density function of fatigue initiation life according to Corollary~\ref{cor:PDF-N-b} with $b=0.95$ using these correlation coefficients and the parameters listed in Table~\ref{table:fit-comparison-Case1}. We plot the PDF in Figure~\ref{fig:fatigue-life-pdf}. The resulting distribution provides valuable insights about the fatigue performance of the printed part under the cyclically applied strain of 0.3\% amplitude. When scheduling maintenance, rather than relying on a point estimate for the fatigue initiation life combined with a conservative empirical safety factor, decision makers can use the PDF values to make informed scheduling choices. Because the PDF of fatigue initiation life $N$ is right-skewed, the expected fatigue initiation life (about $31{,}400$ cycles) is greater than the mode (about $29{,}300$ cycles). In this scenario, compared to scaling down the mean life by a factor of 0.8, using the PDF allows the inspection interval to be scheduled about $4{,}200$ cycles later. This probabilistic approach not only improves decision-making accuracy but also reduces operational costs by saving material and labor.
    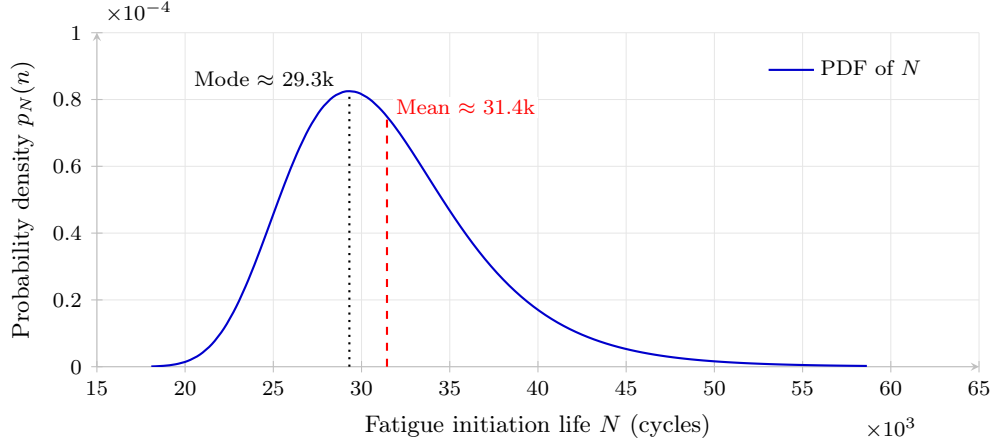
\begin{figure}
        \centering
        \begin{tikzpicture}
  \begin{axis}[
    width=0.85\textwidth,
    height=6cm,
    axis lines=left,
    xlabel={Fatigue initiation life $N$ (cycles)},
    ylabel={Probability density $p_N(n)$},
    xmin=15000, xmax=65000,
    ymin=0, ymax=0.0001,
    scaled x ticks=real:1000,
    scaled y ticks=real:1e-4,
    xtick scale label code/.code={$\times 10^3$},
    ytick scale label code/.code={$\times 10^{-4}$},
    grid=both,
    grid style={line width=.1pt, draw=gray!10},
    major grid style={line width=.2pt, draw=gray!20},
    axis line style={lightgray},
    tick label style={font=\footnotesize},
    label style={font=\small},
    every tick/.style={lightgray},
    legend style={at={(0.95,0.95)}, anchor=north east, font=\footnotesize, draw=none, fill=none}
  ]
    \addplot[
      draw=none,
      forget plot
    ] coordinates {
      (18083.3007, 7.7823454338e-08)
      (18770.7510, 2.5766823271e-07)
      (19458.2014, 7.1674219628e-07)
      (20145.6517, 1.7153119271e-06)
      (20833.1021, 3.6037591371e-06)
      (21520.5525, 6.7616438264e-06)
      (22208.0028, 1.1497374648e-05)
      (22895.4532, 1.7941255012e-05)
      (23582.9036, 2.5972252526e-05)
      (24270.3539, 3.5206300214e-05)
      (24957.8043, 4.5049653046e-05)
      (25645.2546, 5.4798157676e-05)
      (26332.7050, 6.3751745709e-05)
      (27020.1554, 7.1314863031e-05)
      (27707.6057, 7.7063735282e-05)
      (28395.0561, 8.0774187014e-05)
      (29082.5064, 8.2414311192e-05)
      (29769.9568, 8.2112362453e-05)
      (30457.4072, 8.0111855005e-05)
      (31144.8575, 7.6724297866e-05)
      (31832.3079, 7.2286909525e-05)
      (32519.7582, 6.7129329838e-05)
      (33207.2086, 6.1550578038e-05)
      (33894.6590, 5.5805604158e-05)
      (34582.1093, 5.0099732990e-05)
      (35269.5597, 4.4588925007e-05)
      (35957.0101, 3.9383850473e-05)
      (36644.4604, 3.4556087248e-05)
      (37331.9108, 3.0145157353e-05)
      (38019.3611, 2.6165515133e-05)
      (38706.8115, 2.2612939563e-05)
      (39394.2619, 1.9470045940e-05)
      (40081.7122, 1.6710818475e-05)
      (40769.1626, 1.4304186193e-05)
      (41456.6129, 1.2216734327e-05)
      (42144.0633, 1.0414676528e-05)
      (42831.5137, 8.8652218363e-06)
      (43518.9640, 7.5374640443e-06)
      (44206.4144, 6.4029067844e-06)
      (44893.8647, 5.4357200675e-06)
      (45581.3151, 4.6128060224e-06)
      (46268.7655, 3.9137349521e-06)
      (46956.2158, 3.3205983390e-06)
      (47643.6662, 2.8178133666e-06)
      (48331.1165, 2.3919038241e-06)
      (49018.5669, 2.0312746726e-06)
      (49706.0173, 1.7259917851e-06)
      (50393.4676, 1.4675740824e-06)
      (51080.9180, 1.2488021869e-06)
      (51768.3684, 1.0635455330e-06)
      (52455.8187, 9.0660838552e-07)
      (53143.2691, 7.7359424998e-07)
      (53830.7194, 6.6078756105e-07)
      (54518.1698, 5.6505121208e-07)
      (55205.6202, 4.8373834118e-07)
      (55893.0705, 4.1461676931e-07)
      (56580.5209, 3.5580454237e-07)
      (57267.9712, 3.0571513231e-07)
      (57955.4216, 2.6301098043e-07)
      (58642.8720, 2.2656420310e-07)
      (58642.8720, 0)
      (18083.3007, 0)
    };

    \addplot[
      color=blue!80!black,
      thick,
      smooth
    ] coordinates {
      (18083.3007, 7.7823454338e-08)
      (18770.7510, 2.5766823271e-07)
      (19458.2014, 7.1674219628e-07)
      (20145.6517, 1.7153119271e-06)
      (20833.1021, 3.6037591371e-06)
      (21520.5525, 6.7616438264e-06)
      (22208.0028, 1.1497374648e-05)
      (22895.4532, 1.7941255012e-05)
      (23582.9036, 2.5972252526e-05)
      (24270.3539, 3.5206300214e-05)
      (24957.8043, 4.5049653046e-05)
      (25645.2546, 5.4798157676e-05)
      (26332.7050, 6.3751745709e-05)
      (27020.1554, 7.1314863031e-05)
      (27707.6057, 7.7063735282e-05)
      (28395.0561, 8.0774187014e-05)
      (29082.5064, 8.2414311192e-05)
      (29769.9568, 8.2112362453e-05)
      (30457.4072, 8.0111855005e-05)
      (31144.8575, 7.6724297866e-05)
      (31832.3079, 7.2286909525e-05)
      (32519.7582, 6.7129329838e-05)
      (33207.2086, 6.1550578038e-05)
      (33894.6590, 5.5805604158e-05)
      (34582.1093, 5.0099732990e-05)
      (35269.5597, 4.4588925007e-05)
      (35957.0101, 3.9383850473e-05)
      (36644.4604, 3.4556087248e-05)
      (37331.9108, 3.0145157353e-05)
      (38019.3611, 2.6165515133e-05)
      (38706.8115, 2.2612939563e-05)
      (39394.2619, 1.9470045940e-05)
      (40081.7122, 1.6710818475e-05)
      (40769.1626, 1.4304186193e-05)
      (41456.6129, 1.2216734327e-05)
      (42144.0633, 1.0414676528e-05)
      (42831.5137, 8.8652218363e-06)
      (43518.9640, 7.5374640443e-06)
      (44206.4144, 6.4029067844e-06)
      (44893.8647, 5.4357200675e-06)
      (45581.3151, 4.6128060224e-06)
      (46268.7655, 3.9137349521e-06)
      (46956.2158, 3.3205983390e-06)
      (47643.6662, 2.8178133666e-06)
      (48331.1165, 2.3919038241e-06)
      (49018.5669, 2.0312746726e-06)
      (49706.0173, 1.7259917851e-06)
      (50393.4676, 1.4675740824e-06)
      (51080.9180, 1.2488021869e-06)
      (51768.3684, 1.0635455330e-06)
      (52455.8187, 9.0660838552e-07)
      (53143.2691, 7.7359424998e-07)
      (53830.7194, 6.6078756105e-07)
      (54518.1698, 5.6505121208e-07)
      (55205.6202, 4.8373834118e-07)
      (55893.0705, 4.1461676931e-07)
      (56580.5209, 3.5580454237e-07)
      (57267.9712, 3.0571513231e-07)
      (57955.4216, 2.6301098043e-07)
      (58642.8720, 2.2656420310e-07)
    };
    \addlegendentry{PDF of $N$}

    \draw[dashed, red, thick] (axis cs:31443, 0) -- (axis cs:31443, 0.000074)
      node[right=2pt, pos=1.05, font=\footnotesize, text=red, fill=white, inner sep=1pt] {Mean $\approx$ 31.4k};

    \draw[dotted, black, thick] (axis cs:29311, 0) -- (axis cs:29311, 0.000082)
      node[left=2pt, pos=1.05, font=\footnotesize, text=black] {Mode $\approx$ 29.3k};

  \end{axis}
\end{tikzpicture}
        \caption{Case 1: the PDF of the fatigue initiation life $N$ calculated using the analytical ratio distribution of correlated normal variables, see Section~\ref{sec: fatigue-life}. The parameters are fitted from the block extremes of the 64 subcubes, with assumed correlation coefficients $\rho_{G,\gamma} = 0.8$, $\rho_{G,\omega^p} = 0.7$, and $\rho_{\gamma,\omega^p} = 0.6$.}
        \label{fig:fatigue-life-pdf}
    \end{figure}
    
\subsection{Effects of grain size on fatigue initiation life}\label{section:num-results-eff-grain}

In powder bed fusion, local thermal histories affect grain size distributions. In addition, when characterizing grain geometries, measurement noise from EBSD may also introduce uncertainties. Therefore, it is important to study how grain-level variations propagate to the fatigue initiation life. We investigate the effect of grain diameter on the fatigue initiation life by comparing Cases 1, 2, and 3 with varying lognormal distributions of grain diameters while keeping all void parameters fixed, as listed in Table~\ref{table:simulation-cases}. The void volume fraction is kept at $2.0\%$. We perform the same suite of ELAS3D, MicroFract3D, and PRISMS-Plasticity simulations and follow the same distribution-fitting procedure described in Section~\ref{section:num-results}.

\subsubsection{Fitted distributions of $G$, $\gamma$, $\omega^p$}

We summarize the Akaike Information Criterion and Bayesian Information Criterion values calculated using eq.~\eqref{eq:AIC} and eq.~\eqref{eq:BIC} for all distribution fits in Tables~\ref{table:fit-summary-d70} and~\ref{table:fit-summary-d80}. For $G$ in both cases, the generalized extreme value distribution provides the best BIC value, but the normal distribution's BIC value is comparable. For $\gamma$, the normal distribution provides the best fit in both cases. For $\omega^p$, the best fit varies: lognormal for $d = 70\ \mu$m (Case 2) and normal for $d = 80\ \mu$m (Case 3). We select the normal distribution for all three quantities for consistency with the analytical PDF derivation in Section~\ref{section:normal-log-normal}, noting that normal distributions remain competitive fits.

\begin{table}[htbp]
    \centering
    \caption{Case 2: comparison of probability distribution fits for $G$, $\gamma$, and $\omega^p$. Bold values indicate the best (lowest) BIC.}
    \label{table:fit-summary-d70}
    {\footnotesize
    \begin{tabular}{l l l c c}
        \textbf{Quantity} & \textbf{Distribution} & \textbf{Parameters} & \textbf{AIC} & \textbf{BIC} \\
        \hline
        \multirow{3}{*}{$G$}
        & GEV ($\xi=-0.494$) & $\mu=2.27$, $\sigma=0.302$ & 19.59 & \textbf{26.07} \\
        & Normal & $\mu_N=2.33$, $\sigma_N=0.280$ & 22.52 & 26.83 \\
        & Lognormal & $s=0.124$, $\theta=0$, $m=2.32$ & 24.52 & 30.99 \\
        \hline
        \multirow{3}{*}{$\gamma$}
        & GEV ($\xi=-0.419$) & $\mu=5.24$, $\sigma=0.101$ & $-117.52$ & $-111.05$ \\
        & Normal & $\mu_N=5.27$, $\sigma_N=0.0946$ & $-116.17$ & $\mathbf{-111.85}$ \\
        & Lognormal & $s=0.0181$, $\theta=0$, $m=5.26$ & $-113.48$ & $-107.01$ \\
        \hline
        \multirow{3}{*}{$\omega^p$}
        & GEV ($\xi=0.787$) & $\mu=4.05{\times}10^{-4}$, $\sigma=9.82{\times}10^{-5}$ & $-1218.51$ & $-1212.18$ \\
        & Normal & $\mu_N=4.56{\times}10^{-4}$, $\sigma_N=8.74{\times}10^{-5}$ & $-1243.95$ & $-1239.73$ \\
        & Lognormal & $s=0.184$, $\theta=0$, $m=4.48{\times}10^{-4}$ & $-1249.12$ & $\mathbf{-1242.79}$ \\
    \end{tabular}
    }
\end{table}

\begin{table}[htbp]
    \centering
    \caption{Case 3: comparison of probability distribution fits for $G$, $\gamma$, and $\omega^p$. Bold values indicate the best (lowest) BIC.}
    \label{table:fit-summary-d80}
    {\footnotesize
    \begin{tabular}{l l l c c}
        \textbf{Quantity} & \textbf{Distribution} & \textbf{Parameters} & \textbf{AIC} & \textbf{BIC} \\
        \hline
        \multirow{3}{*}{$G$}
        & GEV ($\xi=-0.532$) & $\mu=2.28$, $\sigma=0.349$ & 34.53 & \textbf{41.01} \\
        & Normal & $\mu_N=2.35$, $\sigma_N=0.322$ & 40.42 & 44.73 \\
        & Lognormal & $s=0.145$, $\theta=0$, $m=2.33$ & 48.64 & 55.11 \\
        \hline
        \multirow{3}{*}{$\gamma$}
        & GEV ($\xi=-0.335$) & $\mu=5.24$, $\sigma=0.0954$ & $-119.06$ & $-112.58$ \\
        & Normal & $\mu_N=5.27$, $\sigma_N=0.0911$ & $-120.98$ & $\mathbf{-116.67}$ \\
        & Lognormal & $s=0.0174$, $\theta=0$, $m=5.27$ & $-118.53$ & $-112.05$ \\
        \hline
        \multirow{3}{*}{$\omega^p$}
        & GEV ($\xi=0.787$) & $\mu=4.28{\times}10^{-4}$, $\sigma=1.15{\times}10^{-4}$ & $-1202.81$ & $-1196.47$ \\
        & Normal & $\mu_N=4.80{\times}10^{-4}$, $\sigma_N=8.56{\times}10^{-5}$ & $-1246.47$ & $\mathbf{-1242.25}$ \\
        & Lognormal & $s=0.179$, $\theta=0$, $m=4.72{\times}10^{-4}$ & $-1246.17$ & $-1239.84$ \\
    \end{tabular}
    }
\end{table}

\subsubsection{Fatigue initiation life comparison}

We use the fitted normal distribution parameters and the correlation coefficients $\rho_{G,\gamma} = 0.8$, $\rho_{G,\omega^p} = 0.7$, and $\rho_{\gamma,\omega^p} = 0.6$ to evaluate the analytical PDF of $N$ for each grain size using Corollary~\ref{cor:PDF-N-b} with $b=0.95$. Figure~\ref{fig:fatigue-life-grain-comparison} shows the overlay of the three analytical PDFs. As the mean grain diameter increases from 60 to 80~$\mu$m, the fatigue initiation life distribution shifts slightly to the left, indicating a marginal decrease in fatigue life. This is consistent with the observation that larger grains lead to slightly higher energy release rates $G$ and fatigue indicator parameters $\omega^p$, while the surface energy $\gamma$ remains largely unchanged.

    \begin{figure}[htbp]
        \centering
        \begin{tikzpicture}
  \begin{axis}[
    width=0.85\textwidth,
    height=6cm,
    axis lines=left,
    xlabel={Fatigue initiation life $N$ (cycles)},
    ylabel={Probability density $p_N(n)$},
    ytick scale label code/.code={$\times 10^{-4}$},
    xmin=15000, xmax=65000,
    ymin=0, ymax=10e-05,
    scaled x ticks=real:1000,
    xtick scale label code/.code={$\times 10^3$},
    grid=both,
    grid style={line width=.1pt, draw=gray!10},
    major grid style={line width=.2pt, draw=gray!20},
    axis line style={lightgray},
    tick label style={font=\footnotesize},
    label style={font=\small},
    every tick/.style={lightgray},
    legend style={at={(0.95,0.95)}, anchor=north east, font=\footnotesize, draw=none, fill=none}
  ]
    \addplot[color=blue!80!black, thick, smooth] coordinates {
      (18083.3007, 7.7823454338e-08)
      (18770.7510, 2.5766823271e-07)
      (19458.2014, 7.1674219628e-07)
      (20145.6517, 1.7153119271e-06)
      (20833.1021, 3.6037591371e-06)
      (21520.5525, 6.7616438264e-06)
      (22208.0028, 1.1497374648e-05)
      (22895.4532, 1.7941255012e-05)
      (23582.9036, 2.5972252526e-05)
      (24270.3539, 3.5206300214e-05)
      (24957.8043, 4.5049653046e-05)
      (25645.2546, 5.4798157676e-05)
      (26332.7050, 6.3751745709e-05)
      (27020.1554, 7.1314863031e-05)
      (27707.6057, 7.7063735282e-05)
      (28395.0561, 8.0774187014e-05)
      (29082.5064, 8.2414311192e-05)
      (29769.9568, 8.2112362453e-05)
      (30457.4072, 8.0111855005e-05)
      (31144.8575, 7.6724297866e-05)
      (31832.3079, 7.2286909525e-05)
      (32519.7582, 6.7129329838e-05)
      (33207.2086, 6.1550578038e-05)
      (33894.6590, 5.5805604158e-05)
      (34582.1093, 5.0099732990e-05)
      (35269.5597, 4.4588925007e-05)
      (35957.0101, 3.9383850473e-05)
      (36644.4604, 3.4556087248e-05)
      (37331.9108, 3.0145157353e-05)
      (38019.3611, 2.6165515133e-05)
      (38706.8115, 2.2612939563e-05)
      (39394.2619, 1.9470045940e-05)
      (40081.7122, 1.6710818475e-05)
      (40769.1626, 1.4304186193e-05)
      (41456.6129, 1.2216734327e-05)
      (42144.0633, 1.0414676528e-05)
      (42831.5137, 8.8652218363e-06)
      (43518.9640, 7.5374640443e-06)
      (44206.4144, 6.4029067844e-06)
      (44893.8647, 5.4357200675e-06)
      (45581.3151, 4.6128060224e-06)
      (46268.7655, 3.9137349521e-06)
      (46956.2158, 3.3205983390e-06)
      (47643.6662, 2.8178133666e-06)
      (48331.1165, 2.3919038241e-06)
      (49018.5669, 2.0312746726e-06)
      (49706.0173, 1.7259917851e-06)
      (50393.4676, 1.4675740824e-06)
      (51080.9180, 1.2488021869e-06)
      (51768.3684, 1.0635455330e-06)
      (52455.8187, 9.0660838552e-07)
      (53143.2691, 7.7359424998e-07)
      (53830.7194, 6.6078756105e-07)
      (54518.1698, 5.6505121208e-07)
      (55205.6202, 4.8373834118e-07)
      (55893.0705, 4.1461676931e-07)
      (56580.5209, 3.5580454237e-07)
      (57267.9712, 3.0571513231e-07)
      (57955.4216, 2.6301098043e-07)
      (58642.8720, 2.2656420310e-07)
    };
    \addlegendentry{Case 1 (60 $\mu$m)}

    \addplot[color=red!80!black, thick, smooth, dashed] coordinates {
      (16287.0164, 1.6778756999e-07)
      (17298.8586, 7.3916308263e-07)
      (18310.7007, 2.3551094360e-06)
      (19322.5429, 5.8123221393e-06)
      (20334.3851, 1.1709341758e-05)
      (21346.2272, 2.0052472287e-05)
      (22358.0694, 3.0130153112e-05)
      (23369.9115, 4.0724634483e-05)
      (24381.7537, 5.0504286280e-05)
      (25393.5958, 5.8382881193e-05)
      (26405.4380, 6.3717695612e-05)
      (27417.2802, 6.6333251209e-05)
      (28429.1223, 6.6427384552e-05)
      (29440.9645, 6.4430805648e-05)
      (30452.8066, 6.0873151347e-05)
      (31464.6488, 5.6282224578e-05)
      (32476.4909, 5.1122310998e-05)
      (33488.3331, 4.5765812798e-05)
      (34500.1753, 4.0488291988e-05)
      (35512.0174, 3.5477320478e-05)
      (36523.8596, 3.0847733145e-05)
      (37535.7017, 2.6658403886e-05)
      (38547.5439, 2.2927781649e-05)
      (39559.3860, 1.9646925508e-05)
      (40571.2282, 1.6789714003e-05)
      (41583.0704, 1.4320413380e-05)
      (42594.9125, 1.2199015667e-05)
      (43606.7547, 1.0384815594e-05)
      (44618.5968, 8.8386652314e-06)
      (45630.4390, 7.5242772599e-06)
      (46642.2811, 6.4088706213e-06)
      (47654.1233, 5.4633804220e-06)
      (48665.9654, 4.6623934490e-06)
      (49677.8076, 3.9839227396e-06)
      (50689.6498, 3.4090983637e-06)
      (51701.4919, 2.9218250555e-06)
      (52713.3341, 2.5084384915e-06)
      (53725.1762, 2.1573790021e-06)
      (54737.0184, 1.8588927639e-06)
      (55748.8605, 1.6047648330e-06)
      (56760.7027, 1.3880848137e-06)
      (57772.5449, 1.2030438464e-06)
      (58784.3870, 1.0447604514e-06)
      (59796.2292, 9.0913223739e-07)
      (60808.0713, 7.9271035056e-07)
      (61819.9135, 6.9259363720e-07)
      (62831.7556, 6.0633971323e-07)
      (63843.5978, 5.3189042015e-07)
      (64855.4400, 4.6750944696e-07)
      (65867.2821, 4.1173019298e-07)
      (66879.1243, 3.6331222033e-07)
      (67890.9664, 3.2120489256e-07)
      (68902.8086, 2.8451701325e-07)
      (69914.6507, 2.5249146790e-07)
      (70926.4929, 2.2448403403e-07)
      (71938.3351, 1.9994566254e-07)
      (72950.1772, 1.7840764909e-07)
      (73962.0194, 1.5946921225e-07)
      (74973.8615, 1.4278707657e-07)
      (75985.7037, 1.2806672663e-07)
    };
    \addlegendentry{Case 2 (70 $\mu$m)}

    \addplot[color=green!60!black, thick, smooth, dotted] coordinates {
      (15572.3201, 1.7345291570e-07)
      (16438.8680, 6.6645298774e-07)
      (17305.4158, 1.9727227732e-06)
      (18171.9637, 4.7218899487e-06)
      (19038.5115, 9.4976635815e-06)
      (19905.0594, 1.6554729712e-05)
      (20771.6072, 2.5632312699e-05)
      (21638.1551, 3.5969600607e-05)
      (22504.7029, 4.6503037858e-05)
      (23371.2508, 5.6138838023e-05)
      (24237.7987, 6.3987856262e-05)
      (25104.3465, 6.9498437218e-05)
      (25970.8944, 7.2480091940e-05)
      (26837.4422, 7.3047696993e-05)
      (27703.9901, 7.1526927673e-05)
      (28570.5379, 6.8355118693e-05)
      (29437.0858, 6.3998471804e-05)
      (30303.6336, 5.8893969166e-05)
      (31170.1815, 5.3415715787e-05)
      (32036.7293, 4.7860990782e-05)
      (32903.2772, 4.2449965489e-05)
      (33769.8250, 3.7333536015e-05)
      (34636.3729, 3.2604977819e-05)
      (35502.9208, 2.8312510643e-05)
      (36369.4686, 2.4471042167e-05)
      (37236.0165, 2.1072232300e-05)
      (38102.5643, 1.8092598018e-05)
      (38969.1122, 1.5499719736e-05)
      (39835.6600, 1.3256782983e-05)
      (40702.2079, 1.1325753143e-05)
      (41568.7557, 9.6694816975e-06)
      (42435.3036, 8.2530105755e-06)
      (43301.8514, 7.0442966709e-06)
      (44168.3993, 6.0145324666e-06)
      (45034.9471, 5.1381968775e-06)
      (45901.4950, 4.3929351939e-06)
      (46768.0429, 3.7593388330e-06)
      (47634.5907, 3.2206739107e-06)
      (48501.1386, 2.7625914411e-06)
      (49367.6864, 2.3728402048e-06)
      (50234.2343, 2.0409949886e-06)
      (51100.7821, 1.7582071510e-06)
      (51967.3300, 1.5169806173e-06)
      (52833.8778, 1.3109739198e-06)
      (53700.4257, 1.1348273712e-06)
      (54566.9735, 9.8401358601e-07)
      (55433.5214, 8.5470913419e-07)
      (56300.0692, 7.4368496555e-07)
      (57166.6171, 6.4821327643e-07)
      (58033.1650, 5.6598862592e-07)
      (58899.7128, 4.9506130061e-07)
      (59766.2607, 4.3378114045e-07)
      (60632.8085, 3.8075025462e-07)
      (61499.3564, 3.3478326232e-07)
      (62365.9042, 2.9487388411e-07)
      (63232.4521, 2.6016687951e-07)
      (64098.9999, 2.2993447807e-07)
      (64965.5478, 2.0355658179e-07)
      (65832.0956, 1.8050413026e-07)
      (66698.6435, 1.6032511651e-07)
    };
    \addlegendentry{Case 3 (80 $\mu$m)}

    \draw[solid, blue!80!black, thick] (axis cs:31443, 0) -- (axis cs:31443, 7.4e-05)
      node[above right, font=\small, text=blue!80!black] {Mean $\approx$ 31.4k};
    \draw[dashed, red!80!black, thick] (axis cs:31514, 0) -- (axis cs:31514, 5.7e-05)
      node[above right, font=\small, text=red!80!black, yshift=-1.6cm] {31.5k};
    \draw[dotted, green!60!black, thick] (axis cs:29619, 0) -- (axis cs:29619, 6.4e-05)
      node[above, font=\small, text=green!60!black, yshift=0.2cm] {29.6k};

  \end{axis}
\end{tikzpicture}
        \caption{The effects of grain diameter on fatigue initiation life $N$: comparison of the analytical PDF of the fatigue initiation life $N$ for Cases 1, 2, and 3 listed in Table~\ref{table:simulation-cases}. As the mean grain diameter increases from Case 1 (mean grain diameter 60 $\mu$m, see Table~\ref{table:simulation-cases}) to Case 3 (mean grain diameter 80 $\mu$m), the fatigue initiation life distribution shifts slightly to the left. The correlation coefficients $\rho_{G,\gamma} = 0.8$, $\rho_{G,\omega^p} = 0.7$, $\rho_{\gamma,\omega^p} = 0.6$ are used for all cases. The mean fatigue life calculated using each PDF is approximately $31{,}400$, $31{,}500$, and $29{,}600$ cycles for Case 1, 2, and 3, respectively.}
        \label{fig:fatigue-life-grain-comparison}
    \end{figure}
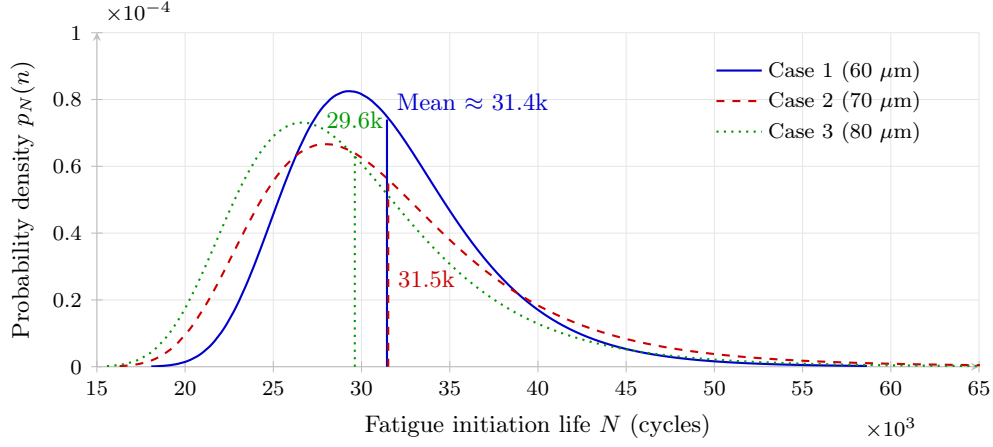

We also directly evaluate $N_i^{0.95} = (2\gamma_i - G_i)/\omega^p_i$ for each subcube using the actual simulation outputs, without fitting distributions. We calculate the per-subcube kernel density estimation (KDE) with the gaussian\_kde function in SciPy v1.18.0~\cite{2020SciPy-NMeth}, where we use Scott's rule to determine the bandwidth factor. We compare the KDE of the per-subcube $N$ values with the analytical PDF for each grain size for all three cases in this study and show the comparison for Case 2 in Figure~\ref{fig:direct-N-grain}, whose kernel bandwidth is approximately $2{,}700$ cycles. For all cases, the per-subcube kernel density estimations and the analytical PDF are in close agreement, validating the assumptions underlying the distribution-fitting approach.

    \begin{figure}[htbp]
        \centering
        \begin{tikzpicture}
  \begin{axis}[
    width=0.85\textwidth,
    height=6cm,
    axis lines=left,
    xlabel={Fatigue initiation life $N$ (cycles)},
    ylabel={Probability density $p_N(n)$},
    ytick scale label code/.code={$\times 10^{-4}$},
    xmin=15000, xmax=65000,
    ymin=0, ymax=10e-05,
    scaled x ticks=real:1000,
    xtick scale label code/.code={$\times 10^3$},
    grid=both,
    grid style={line width=.1pt, draw=gray!10},
    major grid style={line width=.2pt, draw=gray!20},
    axis line style={lightgray},
    tick label style={font=\footnotesize},
    label style={font=\small},
    every tick/.style={lightgray},
    legend style={at={(0.95,0.95)}, anchor=north east, font=\footnotesize, draw=none, fill=none}
  ]
    \addplot[color=red!80!black, thick, smooth] coordinates {
      (16287.0164, 1.6778756999e-07)
      (17298.8586, 7.3916308263e-07)
      (18310.7007, 2.3551094360e-06)
      (19322.5429, 5.8123221393e-06)
      (20334.3851, 1.1709341758e-05)
      (21346.2272, 2.0052472287e-05)
      (22358.0694, 3.0130153112e-05)
      (23369.9115, 4.0724634483e-05)
      (24381.7537, 5.0504286280e-05)
      (25393.5958, 5.8382881193e-05)
      (26405.4380, 6.3717695612e-05)
      (27417.2802, 6.6333251209e-05)
      (28429.1223, 6.6427384552e-05)
      (29440.9645, 6.4430805648e-05)
      (30452.8066, 6.0873151347e-05)
      (31464.6488, 5.6282224578e-05)
      (32476.4909, 5.1122310998e-05)
      (33488.3331, 4.5765812798e-05)
      (34500.1753, 4.0488291988e-05)
      (35512.0174, 3.5477320478e-05)
      (36523.8596, 3.0847733145e-05)
      (37535.7017, 2.6658403886e-05)
      (38547.5439, 2.2927781649e-05)
      (39559.3860, 1.9646925508e-05)
      (40571.2282, 1.6789714003e-05)
      (41583.0704, 1.4320413380e-05)
      (42594.9125, 1.2199015667e-05)
      (43606.7547, 1.0384815594e-05)
      (44618.5968, 8.8386652314e-06)
      (45630.4390, 7.5242772599e-06)
      (46642.2811, 6.4088706213e-06)
      (47654.1233, 5.4633804220e-06)
      (48665.9654, 4.6623934490e-06)
      (49677.8076, 3.9839227396e-06)
      (50689.6498, 3.4090983637e-06)
      (51701.4919, 2.9218250555e-06)
      (52713.3341, 2.5084384915e-06)
      (53725.1762, 2.1573790021e-06)
      (54737.0184, 1.8588927639e-06)
      (55748.8605, 1.6047648330e-06)
      (56760.7027, 1.3880848137e-06)
      (57772.5449, 1.2030438464e-06)
      (58784.3870, 1.0447604514e-06)
      (59796.2292, 9.0913223739e-07)
      (60808.0713, 7.9271035056e-07)
      (61819.9135, 6.9259363720e-07)
      (62831.7556, 6.0633971323e-07)
      (63843.5978, 5.3189042015e-07)
      (64855.4400, 4.6750944696e-07)
      (65867.2821, 4.1173019298e-07)
      (66879.1243, 3.6331222033e-07)
      (67890.9664, 3.2120489256e-07)
      (68902.8086, 2.8451701325e-07)
      (69914.6507, 2.5249146790e-07)
      (70926.4929, 2.2448403403e-07)
      (71938.3351, 1.9994566254e-07)
      (72950.1772, 1.7840764909e-07)
      (73962.0194, 1.5946921225e-07)
      (74973.8615, 1.4278707657e-07)
      (75985.7037, 1.2806672663e-07)
    };
    \addlegendentry{Analytical PDF}
    \addplot[color=red!80!black, thick, dashed, smooth] coordinates {
      (16287.0164, 4.3389689749e-06)
      (17298.8586, 5.8576722648e-06)
      (18310.7007, 7.8344172360e-06)
      (19322.5429, 1.0521913830e-05)
      (20334.3851, 1.4185601370e-05)
      (21346.2272, 1.8998716920e-05)
      (22358.0694, 2.4944298558e-05)
      (23369.9115, 3.1779236259e-05)
      (24381.7537, 3.9086054390e-05)
      (25393.5958, 4.6360735146e-05)
      (26405.4380, 5.3046932735e-05)
      (27417.2802, 5.8509998815e-05)
      (28429.1223, 6.2069835608e-05)
      (29440.9645, 6.3198111976e-05)
      (30452.8066, 6.1796788389e-05)
      (31464.6488, 5.8316695188e-05)
      (32476.4909, 5.3564860062e-05)
      (33488.3331, 4.8331148424e-05)
      (34500.1753, 4.3136932254e-05)
      (35512.0174, 3.8267810227e-05)
      (36523.8596, 3.3950078088e-05)
      (37535.7017, 3.0407918536e-05)
      (38547.5439, 2.7722952275e-05)
      (39559.3860, 2.5678733741e-05)
      (40571.2282, 2.3806897841e-05)
      (41583.0704, 2.1636978512e-05)
      (42594.9125, 1.8945610716e-05)
      (43606.7547, 1.5819664325e-05)
      (44618.5968, 1.2535540608e-05)
      (45630.4390, 9.3938476577e-06)
      (46642.2811, 6.6226939720e-06)
      (47654.1233, 4.3574401501e-06)
      (48665.9654, 2.6490628787e-06)
      (49677.8076, 1.4726687911e-06)
      (50689.6498, 7.4155504200e-07)
      (51701.4919, 3.3553298942e-07)
      (52713.3341, 1.3555080375e-07)
      (53725.1762, 4.8649737368e-08)
      (54737.0184, 1.5452835642e-08)
      (55748.8605, 4.3312720420e-09)
      (56760.7027, 1.0688944687e-09)
      (57772.5449, 2.3186220229e-10)
      (58784.3870, 4.4151048295e-11)
    };
    \addlegendentry{Per-subcube KDE}
    \addplot[only marks, mark=|, mark size=3pt, color=red!80!black!50] coordinates {
      (23926.5332, 0)
      (32640.7949, 0)
      (16650.9330, 0)
      (24402.6736, 0)
      (40669.7553, 0)
      (29441.9186, 0)
      (34223.5661, 0)
      (28224.6586, 0)
      (30785.0705, 0)
      (29367.3059, 0)
      (39572.1799, 0)
      (21126.6041, 0)
      (29778.0746, 0)
      (34271.5848, 0)
      (32614.6376, 0)
      (24746.5690, 0)
      (30979.4250, 0)
      (27073.5267, 0)
      (29811.0267, 0)
      (26632.9979, 0)
      (41450.1430, 0)
      (28053.1119, 0)
      (30172.5316, 0)
      (31812.0082, 0)
      (24086.6253, 0)
      (34875.8932, 0)
      (26047.5773, 0)
      (22335.2750, 0)
      (24931.0392, 0)
      (36564.0705, 0)
      (24656.2393, 0)
      (40752.8542, 0)
      (28447.7772, 0)
      (26779.4541, 0)
      (27488.8165, 0)
      (31230.9409, 0)
      (40480.6058, 0)
      (36603.1583, 0)
      (34001.7846, 0)
      (19844.5107, 0)
      (30781.6578, 0)
      (26604.2736, 0)
      (29783.3294, 0)
      (37485.4462, 0)
      (43599.0821, 0)
      (45876.5734, 0)
      (38943.5139, 0)
      (41511.9402, 0)
      (28794.4456, 0)
      (43817.0952, 0)
      (28399.7360, 0)
      (31059.3723, 0)
      (34544.6375, 0)
      (29655.4194, 0)
      (29565.8930, 0)
      (35426.0596, 0)
      (35542.0629, 0)
      (32112.1184, 0)
      (27595.1917, 0)
      (24475.6857, 0)
      (34504.2740, 0)
    };
    \addlegendentry{Per-subcube $N$}

    \draw[solid, red!80!black, thick] (axis cs:31514, 0) -- (axis cs:31514, 5.6e-05)
      node[right = 2pt, font=\footnotesize, text=red!80!black, fill=white, inner sep=1pt, yshift = 0.2cm] {Mean $\approx$ 31.5k};

    \draw[dashed, red!80!black, thick] (axis cs:31273, 0) -- (axis cs:31273, 5.85e-05)
      node[left=2pt, pos=1.02, font=\footnotesize, text=red!80!black, fill=white, inner sep=1pt, yshift=-2cm] {Mean $\approx$ 31.3k};

  \end{axis}
\end{tikzpicture}
        \caption{Case 2: per-subcube $N$ values compared with the analytical PDF. The analytical PDF is derived from fitted normal distributions; the kernel density estimation is from directly evaluating $N^{0.95} = (2\gamma - G)/\omega^p$ per subcube. The mean life calculated using this PDF is approximately $31{,}500$ cycles. The per-subcube $N$ values are indicated by tick marks on the horizontal axis; the mean life across all subcubes is approximately $31{,}300$ cycles.}
        \label{fig:direct-N-grain}
    \end{figure}
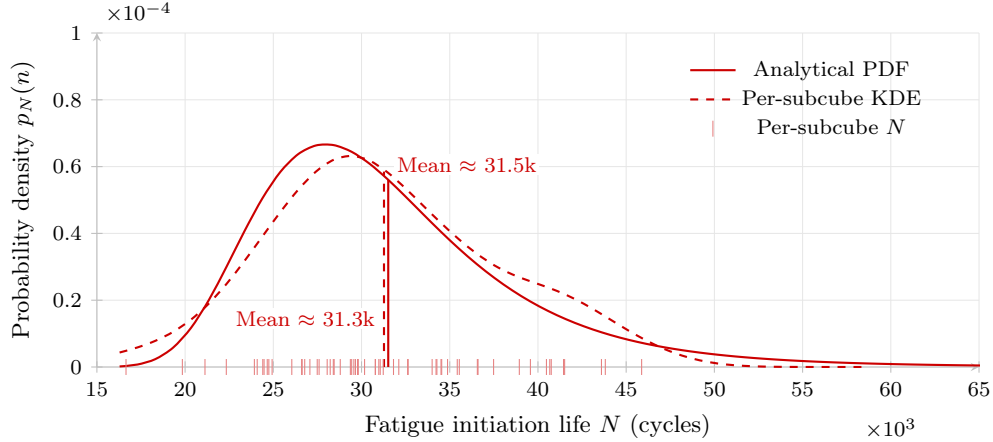

\subsection{Effects of void size on fatigue initiation life}\label{section:num-results-eff-void}

Similar to the previous study, uncertainties in void sizes due to the thermal process during printing and the microCT detections affect the fatigue initiation life. We examine the effect of the void diameter distribution on the fatigue initiation life by comparing Case 4 (void diameter range $[10,80]\ \mu$m) and Case 5 (void diameter range $[50,80]\ \mu$m), which share the same grain parameters but differ in their void diameter range as summarized in Table~\ref{table:simulation-cases}. Case 4 with a smaller minimum void diameter allows for many small voids, while Case 5 with a larger minimum void diameter has fewer but larger voids for the same total void volume~fraction.

\subsubsection{Fitted distributions of $G$, $\gamma$, $\omega^p$}

We summarize the Akaike Information Criterion and Bayesian Information Criterion values calculated using eq.~\eqref{eq:AIC} and eq.~\eqref{eq:BIC} for all distribution fits in Tables~\ref{table:fit-summary-void10} and~\ref{table:fit-summary-void50}. For Case 4, the best BIC fit for $G$ is {GEV}, and $\omega^p$ is best described by a lognormal distribution in both Cases 4 and 5. For $\gamma$, the normal distribution provides the best fit in both cases. As in Section~\ref{section:num-results-eff-grain}, we use normal fits for all three quantities for the analytical PDF evaluation, noting that the normal distribution remains a reasonable approximation despite not always being the best BIC choice.

\begin{table}[htbp]
    \centering
    \caption{Case 4: comparison of probability distribution fits for $G$, $\gamma$, and $\omega^p$. Bold values indicate the best (lowest) BIC.}
    \label{table:fit-summary-void10}
    {\footnotesize
    \begin{tabular}{l l l c c}
        \textbf{Quantity} & \textbf{Distribution} & \textbf{Parameters} & \textbf{AIC} & \textbf{BIC} \\
        \hline
        \multirow{3}{*}{$G$}
        & GEV ($\xi=-0.202$) & $\mu=1.45$, $\sigma=0.196$ & 14.83 & \textbf{21.30} \\
        & Normal & $\mu_N=1.61$, $\sigma_N=0.297$ & 30.34 & 34.66 \\
        & Lognormal & $s=0.176$, $\theta=0$, $m=1.58$ & 23.70 & 30.17 \\
        \hline
        \multirow{3}{*}{$\gamma$}
        & GEV ($\xi=-0.210$) & $\mu=5.06$, $\sigma=0.0571$ & $-173.63$ & $-167.16$ \\
        & Normal & $\mu_N=5.08$, $\sigma_N=0.0600$ & $-174.51$ & $\mathbf{-170.19}$ \\
        & Lognormal & $s=0.0118$, $\theta=0$, $m=5.08$ & $-172.67$ & $-166.19$ \\
        \hline
        \multirow{3}{*}{$\omega^p$}
        & GEV ($\xi=0.781$) & $\mu=3.36{\times}10^{-4}$, $\sigma=5.08{\times}10^{-5}$ & $-1332.82$ & $-1326.39$ \\
        & Normal & $\mu_N=3.67{\times}10^{-4}$, $\sigma_N=4.94{\times}10^{-5}$ & $-1356.72$ & $-1352.43$ \\
        & Lognormal & $s=0.131$, $\theta=0$, $m=3.64{\times}10^{-4}$ & $-1359.23$ & $\mathbf{-1352.80}$ \\
    \end{tabular}
    }
\end{table}

\begin{table}[htbp]
    \centering
    \caption{Case 5: comparison of probability distribution fits for $G$, $\gamma$, and $\omega^p$. Bold values indicate the best (lowest) BIC.}
    \label{table:fit-summary-void50}
    {\footnotesize
    \begin{tabular}{l l l c c}
        \textbf{Quantity} & \textbf{Distribution} & \textbf{Parameters} & \textbf{AIC} & \textbf{BIC} \\
        \hline
        \multirow{3}{*}{$G$}
        & GEV ($\xi=0.772$) & $\mu=2.14$, $\sigma=0.144$ & $-97.03$ & $\mathbf{-90.56}$ \\
        & Normal & $\mu_N=2.16$, $\sigma_N=0.134$ & $-71.71$ & $-67.40$ \\
        & Lognormal & $s=0.0647$, $\theta=0$, $m=2.15$ & $-64.64$ & $-58.16$ \\
        \hline
        \multirow{3}{*}{$\gamma$}
        & GEV ($\xi=0.500$) & $\mu=5.04$, $\sigma=0.0789$ & $-155.29$ & $-148.81$ \\
        & Normal & $\mu_N=5.06$, $\sigma_N=0.0700$ & $-154.71$ & $\mathbf{-150.39}$ \\
        & Lognormal & $s=0.0139$, $\theta=0$, $m=5.06$ & $-152.36$ & $-145.89$ \\
        \hline
        \multirow{3}{*}{$\omega^p$}
        & GEV ($\xi=0.781$) & $\mu=3.72{\times}10^{-4}$, $\sigma=5.58{\times}10^{-5}$ & $-1342.21$ & $-1335.74$ \\
        & Normal & $\mu_N=4.07{\times}10^{-4}$, $\sigma_N=5.65{\times}10^{-5}$ & $-1361.18$ & $-1356.86$ \\
        & Lognormal & $s=0.133$, $\theta=0$, $m=4.03{\times}10^{-4}$ & $-1365.28$ & $\mathbf{-1358.81}$ \\
    \end{tabular}
    }
\end{table}

\subsubsection{Fatigue initiation life comparison {and experimental validation}}

We compare the analytical PDFs of fatigue initiation life for Case 4 and Case 5 in Figure~\ref{fig:fatigue-life-void-comparison}. The microstructure in Case 4 contains many small voids and yields a longer mean fatigue initiation life ($N \approx 40{,}400$ cycles) compared to Case 5 ($N \approx 33{,}600$ cycles). The smaller voids in Case 4 produce lower stress concentrations, resulting in lower $G$ values. This suggests that for a fixed total void volume fraction, having many small voids is more favorable for fatigue life than having fewer, larger voids.

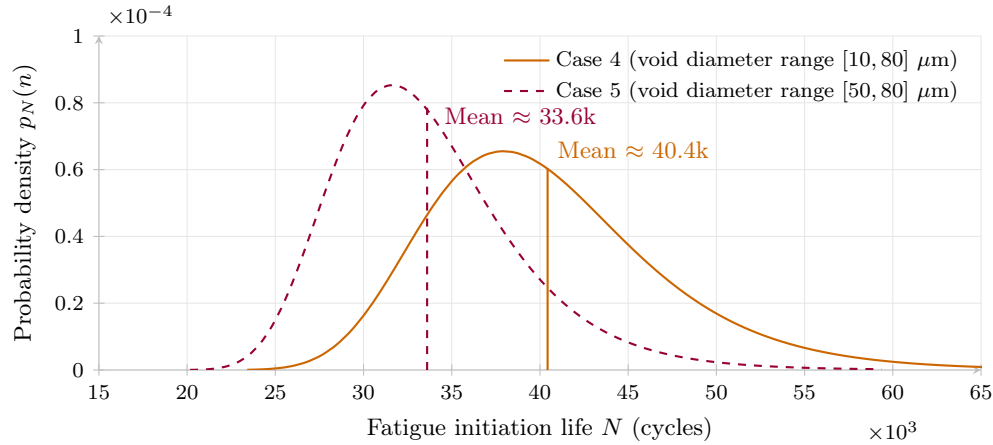
\begin{figure}[htbp]
    \centering
    \begin{tikzpicture}
  \begin{axis}[
    width=0.85\textwidth,
    height=6cm,
    axis lines=left,
    xlabel={Fatigue initiation life $N$ (cycles)},
    ylabel={Probability density $p_N(n)$},
    ytick scale label code/.code={$\times 10^{-4}$},
    xmin=15000, xmax=65000,
    ymin=0, ymax=10e-05,
    scaled x ticks=real:1000,
    xtick scale label code/.code={$\times 10^3$},
    grid=both,
    grid style={line width=.1pt, draw=gray!10},
    major grid style={line width=.2pt, draw=gray!20},
    axis line style={lightgray},
    tick label style={font=\footnotesize},
    label style={font=\small},
    every tick/.style={lightgray},
    legend style={at={(0.99,0.99)}, anchor=north east, font=\footnotesize, draw=none, fill=none}
  ]
    \addplot[color=orange!80!black, thick, smooth] coordinates {
      (23406.4098, 6.3373085131e-08)
      (24248.2528, 1.9275775890e-07)
      (25090.0958, 5.0533873105e-07)
      (25931.9388, 1.1629716775e-06)
      (26773.7818, 2.3868891610e-06)
      (27615.6248, 4.4290405640e-06)
      (28457.4678, 7.5189223274e-06)
      (29299.3108, 1.1799317368e-05)
      (30141.1538, 1.7271490201e-05)
      (30982.9968, 2.3768481783e-05)
      (31824.8398, 3.0965492010e-05)
      (32666.6828, 3.8423998080e-05)
      (33508.5258, 4.5656656200e-05)
      (34350.3688, 5.2196219064e-05)
      (35192.2118, 5.7653639801e-05)
      (36034.0548, 6.1756089951e-05)
      (36875.8978, 6.4362096532e-05)
      (37717.7408, 6.5456318671e-05)
      (38559.5838, 6.5129640138e-05)
      (39401.4268, 6.3551225352e-05)
      (40243.2698, 6.0938530937e-05)
      (41085.1128, 5.7529722554e-05)
      (41926.9558, 5.3561165750e-05)
      (42768.7988, 4.9251074315e-05)
      (43610.6418, 4.4789222023e-05)
      (44452.4848, 4.0331895437e-05)
      (45294.3278, 3.6000925696e-05)
      (46136.1708, 3.1885581841e-05)
      (46978.0138, 2.8046229469e-05)
      (47819.8568, 2.4518866891e-05)
      (48661.6998, 2.1319882364e-05)
      (49503.5428, 1.8450591027e-05)
      (50345.3858, 1.5901288892e-05)
      (51187.2288, 1.3654697503e-05)
      (52029.0718, 1.1688768703e-05)
      (52870.9148, 9.9788804997e-06)
      (53712.7578, 8.4994900266e-06)
      (54554.6008, 7.2253253799e-06)
      (55396.4438, 6.1322010237e-06)
      (56238.2868, 5.1975364891e-06)
      (57080.1298, 4.4006488773e-06)
      (57921.9728, 3.7228787374e-06)
      (58763.8158, 3.1475978343e-06)
      (59605.6589, 2.6601370820e-06)
      (60447.5019, 2.2476639706e-06)
      (61289.3449, 1.8990313161e-06)
      (62131.1879, 1.6046130885e-06)
      (62973.0309, 1.3561382900e-06)
      (63814.8739, 1.1465301876e-06)
      (64656.7169, 9.6975546104e-07)
      (65498.5599, 8.2068582771e-07)
      (66340.4029, 6.9497328882e-07)
      (67182.2459, 5.8893917370e-07)
      (68024.0889, 4.9947652371e-07)
      (68865.9319, 4.2396496885e-07)
      (69707.7749, 3.6019703686e-07)
      (70549.6179, 3.0631474346e-07)
      (71391.4609, 2.6075530326e-07)
      (72233.3039, 2.2220484488e-07)
      (73075.1469, 1.8955908907e-07)
    };
    \addlegendentry{Case 4 (void diameter range $[10,80]\ \mu$m)}

    \addplot[color=purple!80!black, thick, smooth, dashed] coordinates {
      (20170.1291, 3.4745270698e-08)
      (20834.3758, 1.2586771337e-07)
      (21498.6226, 3.8051311145e-07)
      (22162.8693, 9.8341401119e-07)
      (22827.1161, 2.2178082609e-06)
      (23491.3628, 4.4416082591e-06)
      (24155.6096, 8.0187326488e-06)
      (24819.8563, 1.3219917525e-05)
      (25484.1031, 2.0125602022e-05)
      (26148.3498, 2.8566392227e-05)
      (26812.5966, 3.8122841454e-05)
      (27476.8433, 4.8184210931e-05)
      (28141.0900, 5.8046704870e-05)
      (28805.3368, 6.7022718541e-05)
      (29469.5835, 7.4534819401e-05)
      (30133.8303, 8.0177683512e-05)
      (30798.0770, 8.3742756915e-05)
      (31462.3238, 8.5209881887e-05)
      (32126.5705, 8.4715595285e-05)
      (32790.8173, 8.2509289204e-05)
      (33455.0640, 7.8907050991e-05)
      (34119.3107, 7.4250184797e-05)
      (34783.5575, 6.8872335888e-05)
      (35447.8042, 6.3076524914e-05)
      (36112.0510, 5.7121563517e-05)
      (36776.2977, 5.1216276769e-05)
      (37440.5445, 4.5519552278e-05)
      (38104.7912, 4.0144267308e-05)
      (38769.0380, 3.5163424170e-05)
      (39433.2847, 3.0617204021e-05)
      (40097.5314, 2.6520033605e-05)
      (40761.7782, 2.2867094681e-05)
      (41426.0249, 1.9639969872e-05)
      (42090.2717, 1.6811309343e-05)
      (42754.5184, 1.4348528878e-05)
      (43418.7652, 1.2216624756e-05)
      (44083.0119, 1.0380227974e-05)
      (44747.2587, 8.8050319115e-06)
      (45411.5054, 7.4587230711e-06)
      (46075.7522, 6.3115311769e-06)
      (46739.9989, 5.3364976509e-06)
      (47404.2456, 4.5095434066e-06)
      (48068.4924, 3.8093999267e-06)
      (48732.7391, 3.2174526476e-06)
      (49396.9859, 2.7175331273e-06)
      (50061.2326, 2.2956863088e-06)
      (50725.4794, 1.9399312121e-06)
      (51389.7261, 1.6400272800e-06)
      (52053.9729, 1.3872540585e-06)
      (52718.2196, 1.1742085901e-06)
      (53382.4663, 9.9462256690e-07)
      (54046.7131, 8.4319970865e-07)
      (54710.9598, 7.1547278859e-07)
      (55375.2066, 6.0767910056e-07)
      (56039.4533, 5.1665280894e-07)
      (56703.7001, 4.3973247259e-07)
      (57367.9468, 3.7468201322e-07)
      (58032.1936, 3.1962346248e-07)
      (58696.4403, 2.7297993545e-07)
      (59360.6871, 2.3342741801e-07)
    };
    \addlegendentry{Case 5 (void diameter range $[50,80]\ \mu$m)}

    \draw[solid, orange!80!black, thick] (axis cs:40433, 0) -- (axis cs:40433, 6.05e-05)
      node[above right, font=\small, text=orange!80!black] {Mean $\approx$ 40.4k};
    \draw[dashed, purple!80!black, thick] (axis cs:33604, 0) -- (axis cs:33604, 7.85e-05)
      node[above right, font=\small, text=purple!80!black, yshift=-0.35cm, xshift = 0.1cm] {Mean $\approx$ 33.6k};

  \end{axis}
\end{tikzpicture}
    \caption{Comparison of the analytical PDF of the fatigue initiation life $N$ for Case 4 (void diameter range $[10,80]\ \mu$m) and Case 5 (void diameter range $[50,80]\ \mu$m) listed in Table~\ref{table:simulation-cases}. The correlation coefficients $\rho_{G,\gamma} = 0.8$, $\rho_{G,\omega^p} = 0.7$, $\rho_{\gamma,\omega^p} = 0.6$ are used for both cases. The mean fatigue life calculated using each PDF is about $40{,}400$ cycles for Case 4, and $33{,}600$ cycles for Case 5.}
    \label{fig:fatigue-life-void-comparison}
\end{figure}

Using the same approach described in Section~\ref{section:num-results-eff-grain}, we also obtain the KDE of the per-subcube $N$ values for Case 4, whose kernel bandwidth is approximately $2{,}700$ cycles. We compare KDE alongside the analytical PDFs and plot the comparison of Case 4 results in Figure~\ref{fig:direct-N-void}. The agreement between the KDE of per-subcube $N$ values and the analytical PDF is good in both cases, confirming the choice of the normal distribution model and the validity of the distribution-fitting approach for different void~configurations.

\begin{figure}[htbp]
    \centering
    \begin{tikzpicture}
  \begin{axis}[
    width=0.95\textwidth,
    height=6cm,
    axis lines=left,
    xlabel={Fatigue life $N$ (cycles)},
    ylabel={Probability density $p_N(n)$},
    ytick scale label code/.code={$\times 10^{-4}$},
    xmin=15000, xmax=65000,
    ymin=0, ymax=10e-05,
    scaled x ticks=real:1000,
    xtick scale label code/.code={$\times 10^3$},
    grid=both,
    grid style={line width=.1pt, draw=gray!10},
    major grid style={line width=.2pt, draw=gray!20},
    axis line style={lightgray},
    tick label style={font=\footnotesize},
    label style={font=\small},
    every tick/.style={lightgray},
    legend style={at={(0.95,0.95)}, anchor=north east, font=\footnotesize, draw=none, fill=none}
  ]
    \addplot[color=orange!80!black, thick, smooth] coordinates {
      (23406.4098, 6.3373085131e-08)
      (24248.2528, 1.9275775890e-07)
      (25090.0958, 5.0533873105e-07)
      (25931.9388, 1.1629716775e-06)
      (26773.7818, 2.3868891610e-06)
      (27615.6248, 4.4290405640e-06)
      (28457.4678, 7.5189223274e-06)
      (29299.3108, 1.1799317368e-05)
      (30141.1538, 1.7271490201e-05)
      (30982.9968, 2.3768481783e-05)
      (31824.8398, 3.0965492010e-05)
      (32666.6828, 3.8423998080e-05)
      (33508.5258, 4.5656656200e-05)
      (34350.3688, 5.2196219064e-05)
      (35192.2118, 5.7653639801e-05)
      (36034.0548, 6.1756089951e-05)
      (36875.8978, 6.4362096532e-05)
      (37717.7408, 6.5456318671e-05)
      (38559.5838, 6.5129640138e-05)
      (39401.4268, 6.3551225352e-05)
      (40243.2698, 6.0938530937e-05)
      (41085.1128, 5.7529722554e-05)
      (41926.9558, 5.3561165750e-05)
      (42768.7988, 4.9251074315e-05)
      (43610.6418, 4.4789222023e-05)
      (44452.4848, 4.0331895437e-05)
      (45294.3278, 3.6000925696e-05)
      (46136.1708, 3.1885581841e-05)
      (46978.0138, 2.8046229469e-05)
      (47819.8568, 2.4518866891e-05)
      (48661.6998, 2.1319882364e-05)
      (49503.5428, 1.8450591027e-05)
      (50345.3858, 1.5901288892e-05)
      (51187.2288, 1.3654697503e-05)
      (52029.0718, 1.1688768703e-05)
      (52870.9148, 9.9788804997e-06)
      (53712.7578, 8.4994900266e-06)
      (54554.6008, 7.2253253799e-06)
      (55396.4438, 6.1322010237e-06)
      (56238.2868, 5.1975364891e-06)
      (57080.1298, 4.4006488773e-06)
      (57921.9728, 3.7228787374e-06)
      (58763.8158, 3.1475978343e-06)
      (59605.6589, 2.6601370820e-06)
      (60447.5019, 2.2476639706e-06)
      (61289.3449, 1.8990313161e-06)
      (62131.1879, 1.6046130885e-06)
      (62973.0309, 1.3561382900e-06)
      (63814.8739, 1.1465301876e-06)
      (64656.7169, 9.6975546104e-07)
      (65498.5599, 8.2068582771e-07)
      (66340.4029, 6.9497328882e-07)
      (67182.2459, 5.8893917370e-07)
      (68024.0889, 4.9947652371e-07)
      (68865.9319, 4.2396496885e-07)
      (69707.7749, 3.6019703686e-07)
      (70549.6179, 3.0631474346e-07)
      (71391.4609, 2.6075530326e-07)
      (72233.3039, 2.2220484488e-07)
      (73075.1469, 1.8955908907e-07)
    };
    \addlegendentry{Analytical PDF}
    \addplot[color=orange!80!black, thick, dashed, smooth] coordinates {
      (23406.4098, 1.5804967752e-06)
      (24248.2528, 2.7639864876e-06)
      (25090.0958, 4.4850903640e-06)
      (25931.9388, 6.7818627761e-06)
      (26773.7818, 9.6075794857e-06)
      (27615.6248, 1.2837516917e-05)
      (28457.4678, 1.6308168895e-05)
      (29299.3108, 1.9867880156e-05)
      (30141.1538, 2.3405547325e-05)
      (30982.9968, 2.6837078463e-05)
      (31824.8398, 3.0064187473e-05)
      (32666.6828, 3.2951351678e-05)
      (33508.5258, 3.5364249310e-05)
      (34350.3688, 3.7268682720e-05)
      (35192.2118, 3.8829706057e-05)
      (36034.0548, 4.0423026067e-05)
      (36875.8978, 4.2507071951e-05)
      (37717.7408, 4.5393249199e-05)
      (38559.5838, 4.9037001866e-05)
      (39401.4268, 5.2984813284e-05)
      (40243.2698, 5.6526520790e-05)
      (41085.1128, 5.8969707151e-05)
      (41926.9558, 5.9873759283e-05)
      (42768.7988, 5.9121996912e-05)
      (43610.6418, 5.6838127293e-05)
      (44452.4848, 5.3258814185e-05)
      (45294.3278, 4.8671103377e-05)
      (46136.1708, 4.3428161088e-05)
      (46978.0138, 3.7967286056e-05)
      (47819.8568, 3.2752771362e-05)
      (48661.6998, 2.8145534275e-05)
      (49503.5428, 2.4281365000e-05)
      (50345.3858, 2.1047504663e-05)
      (51187.2288, 1.8180012242e-05)
      (52029.0718, 1.5421826178e-05)
      (52870.9148, 1.2648191242e-05)
      (53712.7578, 9.8993980722e-06)
      (54554.6008, 7.3255082028e-06)
      (55396.4438, 5.0940677712e-06)
      (56238.2868, 3.3155714160e-06)
      (57080.1298, 2.0140813140e-06)
      (57921.9728, 1.1391204979e-06)
      (58763.8158, 5.9842902570e-07)
      (59605.6589, 2.9130698117e-07)
      (60447.5019, 1.3107086010e-07)
      (61289.3449, 5.4377116426e-08)
      (62131.1879, 2.0752904142e-08)
      (62973.0309, 7.2708759629e-09)
      (63814.8739, 2.3342437735e-09)
      (64656.7169, 6.8561956882e-10)
      (65498.5599, 1.8400861734e-10)
      (66340.4029, 4.5076991646e-11)
      (67182.2459, 1.0070865188e-11)
      (68024.0889, 2.0506051474e-12)
      (68865.9319, 3.8033870118e-13)
      (69707.7749, 6.4232025734e-14)
      (70549.6179, 9.8737985184e-15)
      (71391.4609, 1.3812063251e-15)
      (72233.3039, 1.7578794853e-16)
      (73075.1469, 2.0352116337e-17)
    };
    \addlegendentry{Per-subcube KDE}
    \addplot[only marks, mark=|, mark size=3pt, color=orange!80!black!50] coordinates {
      (30840.7609, 0)
      (41536.2695, 0)
      (51394.4641, 0)
      (41510.9122, 0)
      (43378.2087, 0)
      (47812.4488, 0)
      (45023.1666, 0)
      (44594.3958, 0)
      (44367.4175, 0)
      (29492.9451, 0)
      (45165.8496, 0)
      (32575.7530, 0)
      (46031.5645, 0)
      (35731.6277, 0)
      (39934.2216, 0)
      (40038.6762, 0)
      (47766.1434, 0)
      (45903.9077, 0)
      (34738.9809, 0)
      (42656.0980, 0)
      (30217.8052, 0)
      (50385.1789, 0)
      (41121.8231, 0)
      (43900.6723, 0)
      (33869.5721, 0)
      (43007.9165, 0)
      (35009.0349, 0)
      (40419.4479, 0)
      (37274.7930, 0)
      (30019.4505, 0)
      (35310.7636, 0)
      (44959.5385, 0)
      (39370.9934, 0)
      (51503.5778, 0)
      (46584.6628, 0)
      (53601.9423, 0)
      (32437.0823, 0)
      (38582.6123, 0)
      (44121.3389, 0)
      (46029.4500, 0)
      (36056.5847, 0)
      (36343.9321, 0)
      (50072.4229, 0)
      (33886.8248, 0)
      (41378.4549, 0)
      (32728.2338, 0)
      (40223.9983, 0)
      (40326.2210, 0)
      (41483.5547, 0)
      (40492.8646, 0)
      (34011.0355, 0)
      (45759.5790, 0)
      (41646.3324, 0)
      (40911.9580, 0)
      (40147.0599, 0)
      (51119.0379, 0)
      (36616.7819, 0)
      (41727.7191, 0)
      (28518.0826, 0)
      (44471.9583, 0)
      (33560.6255, 0)
      (27703.8273, 0)
      (39449.2498, 0)
    };
    \addlegendentry{Per-subcube $N$}

    \addplot[only marks, mark=10-pointed star, mark size=3pt, color=orange!90!black!70] coordinates {
      (36673, 0)
      (36860, 0)
      (49417, 0)
    };
    \addlegendentry{Total life (experiments)}

    \draw[solid, orange!80!black, thick] (axis cs:40483, 0) -- (axis cs:40483, 6.0e-05)
      node[left=1pt, pos=1.02, font=\footnotesize, text=orange!80!black, fill=white, inner sep=1pt, yshift = 0.45 cm] {Mean $\approx$ 40.4k};

    \draw[dashed, orange!80!black, thick] (axis cs:40406, 0) -- (axis cs:40406, 5.75e-05)
      node[right=1pt, pos=1.02, font=\footnotesize, text=orange!80!black, fill=white, inner sep=1pt, yshift = -2.2 cm] {Mean $\approx$ 40.4k}; 

  \end{axis}
\end{tikzpicture}
    \caption{Case 4: per-subcube {fatigue initiation life} $N$ values compared with the analytical PDF. The analytical PDF is derived from fitted normal distributions; the KDE is from directly evaluating $N^{0.95} = (2\gamma - G)/\omega^p$ per subcube. The mean life calculated using this PDF is approximately $40{,}400$ cycles. The per-subcube {fatigue initiation life} $N$ values are indicated by tick marks on the horizontal axis; the mean life across all subcubes is approximately $40{,}400$ cycles. {The total fatigue life values (36{,}673, 36{,}860, and 49{,}417 cycles) for three specimens printed using the parameter corresponding to this case are marked by stars on the horizontal axis.}}
    \label{fig:direct-N-void}
\end{figure}
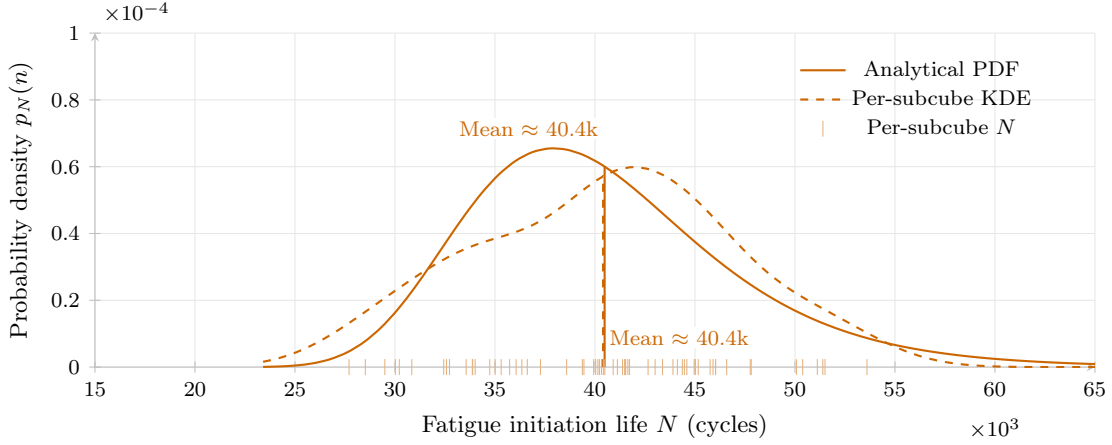

{For Case 4, we validate our framework on experimental results from fatigue lift testing data on three specimens printed using the corresponding lack-of-fusion parameters. The specimens have been machined and polished before mounting to an MTS Bionix servohydraulic test system subject to 0.3\% strain amplitude cyclic loading. The total fatigue life values for the three specimens are  36{,}673, 36{,}860, and 49{,}417 cycles, respectively. As we discuss in Section~\ref{section:fatigue-life-cyclic-loading}, the initiation life accounts for 70\% to 90\% of the total life, these testing results are within the closed-form analytical PDF thus validate our framework of propagating uncertainty from as-printed microstructure features forward to the fatigue initiation life.}

\section{Conclusion and future work} \label{sec: con}

We have developed and demonstrated an integrated framework for quantifying the uncertainty of fatigue initiation life in laser powder bed fusion 3D-printed parts. By processing microstructure data from electron backscatter diffraction and micro-computed tomography scan data of a printed specimen using three physical models, we derive a closed-form, analytical PDF for fatigue initiation life. The models we use and the fatigue-affecting quantities we obtain are the linear elastic finite element analysis ELAS3D for the elastic energy release rate $G$, the graph-theoretic crack path simulation MicroFract3D for the surface energy $\gamma$, and the crystal plasticity finite element analysis PRISMS-Plasticity for the fatigue indicator parameter $\omega^p$. We speed up the ELAS3D simulation by implementing a two-fidelity scheme with a suitable preconditioner and warm-start initial guess. We integrate our two-phase material modeling into the MicroFract3D solver and consider proper material boundary interactions. We use the same microstructures for all three models to obtain the fatigue-affecting quantities, so we consider correlations among these quantities and propagate uncertainties from the as-built microstructure to the fatigue initiation life. 

Using the propagation framework, we show that the resulting probability density function of the fatigue initiation life can improve the accuracy of the decision on scheduling inspection intervals and save operational costs. Through careful design of numerical experiments, we further investigate the sensitivity of the fatigue initiation life distribution to microstructure parameters through parametric studies on grain size and void size. Increasing the mean grain diameter from 60 to 80~$\mu$m leads to a modest decrease in the mean fatigue life. For voids, the minimum void diameter impacts fatigue life. For a fixed total void volume fraction, configurations with many small voids yield longer fatigue lives than those with fewer, larger voids. In all cases, we validate the analytical PDF against direct per-subcube evaluations of the fatigue life formula.

We identify several directions to extend this work: (1) Using PETSc solvers for ELAS3D to further accelerate elastic stress calculations and developing surrogate models for the crystal plasticity finite element simulations to bypass the high computational costs of the PRISMS-Plasticity solver, (2) calibrating the cross-correlation coefficients directly from experimental fatigue life testing of failed prints, and (3) integrating microstructure prediction models with the calibrated fatigue life distribution to enable the prediction of fatigue reliability directly from print parameters.

\section*{CRediT authorship contribution statement}
\textbf{Yulin Guo}: Conceptualization, Methodology, Data Curation, Software, Formal analysis, Investigation, Software, Visualization, Writing - original draft, Writing - review \& editing; \textbf{Boris Kramer}: Conceptualization, Formal analysis, Investigation, Writing - review and editing, Funding acquisition, Project administration, Supervision; \textbf{Veera Sundararaghavan}: Conceptualization, Writing - review \& editing, Supervision,
Project administration, Funding acquisition. 

\section*{Acknowledgment}

We thank our collaborators for providing EBSD and microCT figures {and experimental results}. Figures~\ref{fig:printed-samples}, {the printed specimen photo and EBSD scans on XY, XZ, and YZ planes in Figure~\ref{fig:pipeline}} are produced by Mohsen Taheri Andani, Texas A\&M University; {the void visualization in Figure~\ref{fig:pipeline} and the total experimental fatigue life values in Figure~\ref{fig:direct-N-void} are} produced using data provided by Shuai Shao, Auburn University. We also thank Juyoung Jeong from University of Michigan, Ann Arbor for sharing the ELAS3D simulation codes and input files for PRISMS-Plasticity.

Funding information: Y. Guo, B. Kramer, and V. Sundararaghavan were financially supported by the Defense Advanced Research Projects Agency (DARPA) Cooperative Agreement No. HR0011-25-2-0009, "Predictive Real-time Intelligence for Metallic Endurance (PRIME)."

\section*{Conflict of interest}
All authors declare that they have no known competing financial interests or personal relationships that could have appeared to influence the work reported in this paper.

\section*{Data availability and replication of results}
Input files for all simulations used to generate the results are available to download from the Github repository: https://github.com/yulin-g/UQ-fatigue-life-PBF-AM. The results presented in this manuscript can be replicated using the provided material in the Github repository.

\section*{Tool and computational resource disclosure}
During the preparation of this work, Y. Guo used ChatGPT, Claude, and Gemini models to assist with improving the clarity of existing sentences, troubleshooting programming-related issues, and converting matplotlib figures to TikZ formats. All suggestions generated by these tools were independently verified by Y. Guo. The authours reviewed and edited the content as needed as take full responsibility for the content of the published article.

\bibliography{references-R1}
\bibliographystyle{acm}

\end{document}